%% file: heat_asymptotic_corrected_Jens.tex
\documentclass[12pt]{iopart}

\usepackage{iopams}
\usepackage[T1]{fontenc}
\usepackage{dsfont}
\usepackage{amsopn}
\usepackage{setstack}
\usepackage[german,british]{babel} 
\usepackage{graphics,bbm,url,amsthm,color}
\usepackage{comment}
\input{def}

\begin{document}
%

\title{Heat-kernel and resolvent asymptotics for Schr\"odinger operators on metric graphs}
\author{Jens Bolte, Sebastian Egger and Ralf Rueckriemen}
\address{Department of Mathematics, Royal Holloway, University of London, Egham, TW20 0EX, UK}
\ead{jens.bolte@rhul.ac.uk and sebastian.egger@rhul.ac.uk and ralf@rueckriemen.de }
\begin{abstract}
We consider Schr\"odinger operators on compact and non-compact (finite) metric graphs.
For such operators we analyse their spectra, prove that their resolvents can be
represented as integral operators and introduce trace-class regularisations of the
resolvents. Our main result is a complete asymptotic expansion of the trace of the
(regularised) heat-semigroup generated by the Schr\"odinger operator. We also determine 
the leading coefficients in the expansion explicitly.
\end{abstract}
\hspace*{2.5cm}{\small{\bf Keywords:} Quantum graph, Schr\"odinger operator, resolvent, 
heat kernel}\\
\hspace*{2.5cm}{\small{\bf Mathematical Subject Classicication:} 81Q35; 81Q10; 58J50; 47A10}
\maketitle
\section{Introduction}
A quantum graph is a metric graph with a differential operator acting on functions defined
on the edges of the graph. (Although one may also consider infinite networks, in the following
we shall always focus on graphs with finitely many edges and vertices.) Often the quantum graph 
operator is taken to be a Laplacian, $-\Delta$. More generally, however, Schr\"odinger 
operators of the form $H=-\Delta+V$ with a potential $V$ defined on the edges are of interest. 
Quantum graph models are used in many applications where a one-dimensional (wave-) motion in 
a structure with non-trivial connectivity is studied, see \cite{Gnutzman:2006,AGAProc:2008} 
for reviews. They are also used as models in quantum chaos \cite{KottosSmilansky:1998} and 
bear many similarities with problems in spectral geometry. In the latter area heat-trace 
asymptotics are an important tool to gain geometric and topological information on a manifold 
from the spectral data of a suitable operator. The same is achieved in quantum graph models 
using Laplace operators.

Spectral properties of quantum graph Laplacians are well understood. This is mainly due
to the fact that exact trace formulae exist 
\cite{Roth:1984a,KottosSmilansky:1998,Schrader:2007,BE:2008}. They express spectral 
functions of the Laplacian in terms of a sum over periodic walks on the graph. Often 
the spectral function will be the trace of the heat-semigroup, so that the trace formula 
implies a complete asymptotic expansion for the heat trace as well as an exponential bound 
for the remainder.

The spectral information available for Schr\"odinger operators $H=-\Delta+V$ on graphs is 
less detailed. This can be seen from the fact that a trace formula for $H$ is only known
for spectral functions supported away from low eigenvalues \cite{Rueckriemen:2012}.
The determination of the contribution of low eigenvalues is complicated by the fact that the 
potential may lead to trapped orbits, corresponding to resonances of $H$. Therefore, independent
studies of heat-trace asymptotics will complement the spectral information for Schr\"odinger
operators on graphs. Such heat-kernel expansions were first determined in \cite{Ralf:2012a}
for Schr\"odinger operators on compact graphs under certain (rather stringent) conditions 
on the behaviour of the potential in the vertices. The method used in \cite{Ralf:2012a}
is based on a parametrix construction for the heat-semigroup, which requires sufficient
regularity of the symbol of $H$ in the singular points of the graph (i.e., the vertices).

In this paper we generalise the results of \cite{Ralf:2012a} in two major ways: We drop 
the conditions on the behaviour of the potential in the vertices and we allow for external 
edges, turning the graph into a non-compact graph. This is possible since our approach
does not rely on parametrices but rather uses resolvent kernels very much in the spirit
of \cite{KostrykinSchrader:2006b,Schrader:2007}, where the standard Laplacian is considered. 
We also allow for general, self-adjoint boundary conditions in the vertices, while in 
\cite{Ralf:2012a} only the case of Kirchhoff conditions was considered. 

The paper is organised as follows: In Section~\ref{sec:QG} we review the construction of
quantum graphs. The following Section~\ref{sec:pspec} is devoted to an analysis of 
the point spectrum of $H$ and the associated eigenfunctions. Here we introduce a
secular equation that allows to characterise eigenvalues. The secular equation involves 
a matrix $\mf{S}(k)$ that encodes the boundary conditions in the vertices. In 
Section~\ref{sec:res} we represent the resolvent of $H$ as an integral operator and 
express its kernel in terms of $\mf{S}(k)$. We also introduce a regularisation of the
resolvent that in the case of a non-compact graph leads to a trace-class operator. 
Asymptotic expansions of the matrix $\mf{S}(k)$ are developed in Section~\ref{sec:asy},
and further asymptotic expansions are performed in Section~\ref{asyml}. A complete 
asymptotic expansion for the trace of the regularised resolvent is proven in 
Section~\ref{l2}. The result is presented in Theorem~\ref{151}. Our main result is
contained in Section~\ref{sec:heatasy}. In Theorem~\ref{49} we prove a complete asymptotic
expansion for the trace of a regularised heat-semigroup generated by a Schr\"odinger 
operator on a general (non-compact) metric graph.
\section{Quantum graphs} 
\label{sec:QG}
A metric graph $\Gamma$ is a finite, connected, combinatorial graph with a metric structure. 
It consists of a finite set $\mc{V}$ of vertices and a finite set 
$\mc{E}=\mc{E}_{\inte}\cup\mc{E}_{\ex}$ of edges. Edges $e\in\mc{E}$ are either internal,
$e\in\mc{E}_{\inte}$, or external, $e\in\mc{E}_{\ex}$. Internal edges link two vertices,
which are identified with the edge ends, and external edges are connected to a single vertex.
We set $\mf{V}:=|\mc{V}|$, $E_{\inte}:=\left|\mc{E}_{\inte}\right|$, $E_\ex:=\left|\mc{E}_{\ex}\right|$ 
and $E:=E_{\ex}+2E_{\inte}$. When an edge end is connected to a vertex we say that the 
edge $e$ is adjacent to the vertex $v$, denoted as $e\sim v$. The number of edges adjacent
to a vertex $v$ is its degree $d_v$.  Tadpoles, i.e., internal edges that are only adjacent to a 
single vertex, are allowed. However, this case will later become irrelevant when we introduce
additional vertices.

A metric structure is defined by assigning intervals to edges. Each internal edge $e\in\mc{E}_{\inte}$ 
is assigned an interval $I_e=[0,l_e]$ of finite length $l_e$, whereas each external edge is assigned 
a half-infinite interval $I_e=[0,\infty)$. For convenience we then sometimes write $l_e=\infty$. 
We also introduce the vector $\bsy{l}:=(l_1,\dots,l_{E_{\inte}})^T\in\rz_+^{E_{\inte}}$ of (finite) 
edge lengths. The volume $\mathcal{L}$ of the interior part of the metric graph is 
$\mathcal{L}:=\sum\limits_{e\in\mc{E}_{\inte}}l_e$. Given two points $x,y$ on $\Gamma$, 
a path from $x$ to $y$ is a succession of edges, connected in vertices, such that $x$ is
on the initial edge and $y$ is on the final edge. The distance $d(x,y)$ of the points is the
minimum of the lengths of all paths from $x$ to $y$.

Functions on $\Gamma$ are collections of functions on the intervals associated with 
edges, so that we can introduce the quantum graph Hilbert space
\begin{equation}
\label{1}
L^2(\Gamma)=\bigoplus_{e\in\mc{E}}L^2(0,l_e).
\end{equation}
Similarly, other function spaces such as Sobolev spaces $H^m(\Gamma)$
and spaces of smooth functions are defined. A Schr\"odinger operator is a linear
operator on a dense domain $\mc{D}\subset L^2(\Gamma)$, acting on a function on 
edge $e$ as
\begin{equation}
\label{2}
(H\psi)_e:=-\psi_e'' + V_e\psi_e.
\end{equation}
Here $\psi_e'(x)=\frac{\ud\psi_e}{\ud x}(x)$, $x\in(0,l_e)$, and 
\begin{equation}
\label{3}
V_e\in C^{\infty}(0,l_e),\quad e\in\mc{E}_{\inte},\quad\mbox{or}\quad V_e\in 
C^{\infty}_0[0,\infty),\quad e\in\mc{E}_{\ex},
\end{equation}
is a potential on the edge $e$, where the second requirement means that the potential 
has a compact support on an external edge but need not  vanish at the vertex. Hence,
\begin{equation}
H=-\Delta+V,
\end{equation}
where $V$ is to be understood as a diagonal matrix with entries $V_e$ as in \eref{3}
on the diagonal. 

In order to determine domains of self-adjointness for the Schr\"odinger operator $H$
the Laplacian $-\Delta$ has to be realised as a self-adjoint operator on a suitable 
domain $\mc{D}$. Then $H$ will be self-adjoint on the same domain.
Classifications of self-adjoint realisations of the Laplacian are well known 
\cite{KostrykinSchrader:1999,Kuchment:2004}. They require boundary values of functions 
and their (inward) derivatives,
\begin{equation}
\label{5}
\eq{
\underline{\psi}&=\left(\left\{\psi_e(0)\right\}_{e\in\mc{E}_{\ex}},\left\{\psi_e(0)
                  \right\}_{e\in\mc{E}_{\inte}},\left\{\psi_e\lk l_e\rk\right\}_{e\in\mc{E}_{\inte}}
                  \right)^T\in\kz^{E},\\
\underline{\psi'}&=\left(\left\{\psi'_e(0)\right\}_{e\in\mc{E}_{\ex}},\left\{\psi'_e(0)
                   \right\}_{e\in\mc{E}_{\inte}},\left\{-\psi_e'\lk l_e\rk\right\}_{e\in\mc{E}_{\inte}}
                   \right)^T\in\kz^{E},
}
\end{equation}
which are well defined for $\psi\in H^2(\Gamma)$. Following \cite{Kuchment:2004},
a parametrisation of self-adjoint realisations can be achieved in terms of orthogonal
projectors $P$ on $\kz^E$ and self-adjoint maps $L$ on $\kz^E$ satisfying
$P^{\bot}LP^{\bot}=L$: Every self-adjoint realisation of the Laplacian has a unique
representation of the form
\begin{equation}
\label{4}
\mc{D}(P,L):=\lgk\psi\in H^2(\Gamma);\quad (P+L)\underline{\psi}+P^{\bot}\underline{\psi'}
=0\rgk.
\end{equation}
The boundary conditions imposed on functions in $\mc{D}(P,L)$ via \eref{4} do not
necessarily reflect the connectivity of the graph. This will only be the case if
\begin{equation}
\label{localbc}
P=\bigoplus_{v\in\mc{V}}P_v\quad\text{and}\quad L=\bigoplus_{v\in\mc{V}}L_v
\end{equation}
in such a way that $P_v,L_v$ act on the space $\kz^{d_v}$ of boundary values related to 
the $d_v$ edge ends adjacent to $v$. We call such boundary conditions {\it local}. The focus 
on local boundary conditions becomes important in the Sections \ref{l2} and \ref{sec:heatasy} 
where we derive the first two leading contributions of the resolvent kernel and the heat kernel 
for two points that are distinct on $\Gamma$ but possess zero distance (ends of different edges
connected in a vertex).

Examples of local boundary conditions would be Kirchhoff (or standard) conditions. 
Introducing coordinates such that $v$ corresponds to $x=0$ on every edge adjacent
to $v$, this means $\psi_e(0)=\psi_{e'}(0)$ if $e,e'$ are both adjacent to $v$, and
\begin{equation}
\label{4a}
\sum_{e\in\mc{E},\atop e\sim v}\psi'_e(0)  =0.
\end{equation}
Notice that a vertex $v$ with degree $d_v=2$ and with Kirchhoff conditions (\ref{4a}) 
can be removed since the functions and their derivatives are continuous across the
vertex.

The same fact can be used to add vertices of degree two with Kirchhoff conditions without 
changing the operator $H$. This allows us to make a few simplifying assumptions without 
losing generality: We first exclude potentials on external edges. When $e\in\mc{E}_{\ex}$ 
with non-vanishing $V_e\in C_0^\infty(0,\infty)$ we add a vertex  $v$ of degree $d_v=2$ 
on $e$ outside of the support of $V_e$. Secondly, we remove tadpoles by introducing an 
additional vertex on that edge. This procedure may change the underlying graph, but not 
the operator $H$.

The above conventions allow us to denote the boundary values of edge-potentials as 
$V_e(v)$ when $e\sim v$. Furthermore, the outward derivative of $V_e$ at $v$ is denoted 
as $V'(v)_e$, i.e., either $V'(v)_e =V_e'(0)$ or $V'(v)_e =-V_e'(l_e)$, depending on 
which edge end is adjacent to $v$.
\section{Eigenvalues and eigenfunctions}
\label{sec:pspec}
An eigenfunction of the Schr\"odinger operator $H$ is a function 
$\varphi=\{\varphi_e\}_{e\in\mc{E}}\in\mc{D}(P,L)$ such that there exists $\lambda\in\rz$ 
with
\begin{equation}
\label{5a}
H\varphi=\lambda\varphi.
\end{equation}
It is convenient to set $\lambda=k^2$ with $k\in\kz$. Self-adjointness of $H$ then implies 
that either $k\in\rz$, when $\lambda\geq 0$, or $k=\ui\kappa$ with $\kappa\in\rz$, when
$\lambda<0$. 

On an external edge $e\in\mc{E}_{\ex}$, where $V_e=0$, a fundamental system of solutions 
is given by the functions $\ue^{ikx}$ and $\ue^{-ikx}$ when $k^2\neq 0$. The condition 
$\varphi\in L^2(\Gamma)$ then implies for $k\in\rz$ (i.e., $\lambda>0$) that the eigenfunction
$\phi$ has to vanish on external edges. When 
$\lambda<0$ only one of the functions is permitted, depending on the sign of $\im k$.
We here always choose
\begin{equation}
\label{5c}
\varphi_e(x)=c_e\,\ue^{ikx},\quad\im k>0,\quad c_e\in\kz.
\end{equation}
On an internal edge $e\in\mc{E}_{\inte}$ the equation implied by \eref{5a} reads
\begin{equation}
\label{5b}
-\varphi_e'' +V_e\varphi_e-k^2 \varphi_e =0,
\end{equation}
where $k\in\kz$ such that $k^2=\lambda$. We shall need a certain type of fundamental
solutions on the internal edges. Below two more specific classes will be used in \eref{inival} 
and in Lemma~\ref{6}.
\begin{defn}
\label{def:fundsys}
A pair $\{u_e^+(k;\cdot),u_e^-(k;\cdot)\}$ of functions in $C^\infty(I_e)$ is said to
be a system of admissible fundamental solutions on an internal edge $e\in\mc{E}_{\inte}$, 
if the functions $u_e^\pm(k;x)$ are solutions of the equation \eref{5b}, are analytic 
in $k\in S_\delta$, where
\be
\label{SML}
S_{\delta}:=\lgk z\in\kz; \quad 0<\left|z\right|<\infty, \quad \left|\arg(-z)\right|
>\delta\rgk, 
\ee
for some (small) $\delta>0$ and satisfy the condition
\be
\label{connection}
\overline{u_e^+(k;x)}=u_e^-(\bar k;x),
\ee
for all $k\in S_\delta$. We also fix the normalisation $u_e^+(k;0)=1=u_e^-(k;0)$.
\end{defn}
\begin{rem}
\label{non_vanishing_of_fund_solutions}
Set 
\begin{equation}
u^+_e(k;x) = r_e(k;x)\,\ue^{\ui\phi_e(k;x)},
\end{equation}
with smooth, real-valued functions $r_e(k;\cdot)$ and 
$\phi_e(k;\cdot)$, such that $r_e(k;\cdot)$ is non-negative.
Then 
 \begin{equation}
u^-_e(k;x) = r_e(\overline{k};x)\,\ue^{-\ui\phi_e(\overline{k};x)}.
\end{equation}
If $k$ is real, the Wronskian of this fundamental system takes the form
\begin{equation}
\label{def:Wronski}
W_e(k) = {u^+_e}'(k;x)u^-_e(k;x)-u^+_e(k;x){u^-_e}'(k;x) = 2\ui{\phi_e}'(k;x)\,r_e(k;x)^2.
\end{equation}
This relation, in particular, implies that $u^\pm_e(k;x)\neq 0$ and ${u^\pm_e}'(k;x)\neq 0$
for all $k\in\rz$ and all $x\in I_e$. As the functions $u^\pm_e(k;\cdot)$ are analytic
in $k\in S_\delta$, this also implies that they, and their derivatives, are non-zero for 
all $x\in I_e$ with $k$ in a neighbourhood of the positive half-line.
\end{rem}
Note that these conditions do not uniquely determine a system of admissible fundamental 
solutions. Examples of admissible fundamental solutions can be found in \cite{Fedoryuk:1993}
as well as in \cite{Pöschl:1987}. In the former case the functions possess asymptotic
expansions in $k$, a fact that we shall use below. 

In order to characterise eigenvalues and eigenfunctions \eref{5a} of $H$ we follow 
the method devised in \cite{KostrykinSchrader:2006b}. For this we need to introduce 
the following matrices, using any (fixed) system of admissible fundamental solutions.
\begin{equation}
\label{16a}
\eq{
X(k;\bsy{l})&:=\bma{ccc}
\eins & 0 & 0\\
0 & \eins & \eins\\
0 & \bsy{u}_+(k;\bsy{l}) & \bsy{u}_-(k;\bsy{l})\\
\ema,\\
Y(k;\bsy{l})&:=\bma{ccc}
\ui k\eins & 0 & 0 \\
0 & \bsy{u}_+'(k;\bsy{0}) & \bsy{u}_-'(k;\bsy{0})\\
0 & -\bsy{u}_+'(k;\bsy{l}) & -\bsy{u}_-'(k;\bsy{l})
\ema,
}
\end{equation}
where $\bsy{u}_\pm(k;\bsy{x})$ are diagonal $E_{\inte}\times E_{\inte}$
matrices with diagonal entries $u^\pm_e(k;x_e)$. We then define
\begin{equation}
\label{neue1}
Z(k;P,L,\bsy{l}):=(P+L)X(k;\bsy{l})+P^{\bot}Y(k;\bsy{l}),
\end{equation}
which can be used to set up a characteristic equation.
\begin{lemma}
\label{14}
Let $k^2 \neq 0$ be an eigenvalue of the Schr\"odinger operator $H$. Then one
can choose $k\in S_\delta$, and for this choice
\begin{equation}
\label{15}
\det Z(k;P,L,\bsy{l})=0.
\end{equation}
Furthermore, the  pure point spectrum $\sigma_{pp}(H)$ of $H$ consists of eigenvalues of
finite multiplicities, bounded by $E$, and has no finite accumulation 
point.
\end{lemma}
\begin{proof}
The proof follows \cite{KostrykinSchrader:2006b}. For every $k\in S_{\delta}$ and every 
$e\in\mc{E}_{\inte}$ the functions $u^+_e(k;\cdot)$ and $u^-_e(k;\cdot)$ that are used
to define $Z$ form a complete system of solutions for \eref{5b}. Moreover, one can 
choose $\delta$ such that every eigenvalue $k^2\neq 0$ has a root $k\in S_\delta$,
and either $k \in \mathbb{R}$ or $k$ is of the form $k=\ui\kappa$, $\kappa>0$. Together 
with \eref{5c} this implies that an eigenfunction \eref{5a} can be represented as
\be
\label{17}
\varphi_e (x) =
\cases{
\gamma_e\ue^{\ui k x}, & $e \in\mc{E}_{\ex}$,\\
\alpha_e u^+_e(k;x)+\beta_e u^-_e(k;x), & $e \in \mc{E}_{\inte}$,                  
}
\ee
where $\alpha_e,\beta_e,\gamma_e$ are complex coefficients. The boundary condition
implied by \eref{4} can be rearranged to yield
\begin{equation}
\label{Z4}
Z(k;P,L,\bsy{l}) \bma{c}\boldsymbol\gamma\\ \bsy\alpha \\ \bsy\beta\ema=0,
\end{equation}
where $\bsy{\alpha},\bsy{\beta},\bsy{\gamma}$ are vectors with entries 
$\alpha_e,\beta_e,\gamma_e$, respectively. Hence, every eigenvalue $k^2$ of $H$
leads to a zero of $\det Z(k;P,L,\bsy{l})$. 

Conversely, every zero of $\det Z(k;P,L,\bsy{l})$ is associated with a non-trivial
solution vector $\bsy{\alpha},\bsy{\beta},\bsy{\gamma}$, from which a function
\eref{17} can be constructed that is in $C^\infty(\Gamma)$ and satisfies the vertex
conditions. If $k=\ui\kappa$, $\kappa>0$, this function is in $L^2(\Gamma)$ and thus
is an eigenfunction of $H$, corresponding to the eigenvalue $k^2<0$. When $k^2>0$, 
however, the function \eref{17} is in $L^2(\Gamma)$, iff $\gamma=\bsy{0}$. Thus the 
multiplicities of the eigenvalues are bounded by $E$.

Due to the assumptions made in Definition~\ref{def:fundsys} the matrix entries of
$Z(k;P,L,\bsy{l})$ are analytic in $k\in S_\delta$, hence the same holds for $\det Z(k)$. 
Since $\det Z(k)$ is not identically zero, its zeros in $S_\delta$ form a countable set and 
do not have an accumulation point in $S_\delta$. Hence the set of non-zero eigenvalues is
countable and has no finite, non-zero accumulation point.
\end{proof}
We denote the countable subset of $k\in S_\delta$ for which $Z(k;P,L,\bsy{l})$ is not
invertible as
\begin{equation}
\Sigma_Z :=\{ k\in S_\delta;\  \det Z(k;P,L,\bsy{l})=0\}.
\end{equation}
Using this notation, Lemma~\ref{14} states that 
$\sigma_{pp}(H)\subseteq\{k^2\in\rz;\ k\in\Sigma_Z\}\cup\{0\}$.

In the case of the Laplacian one often uses an alternative characteristic equation for
its eigenvalues involving an $\mf{S}$-matrix, see 
\cite{KottosSmilansky:1998,KostrykinSchrader:2006b}. For a Schr\"odinger operator we now
define an analogous quantity,
\begin{equation}
\label{71}
\mathfrak{S}(k;P,L):=-\lk P+L+P^{\perp}\overline{D\lk\overline{k}\rk}\rk^{-1}
\lk P+L+P^{\perp}D(k)\rk,
\end{equation}
for $k\in S_\delta$, where
\be
\label{56}
D(k):=R_2(k;\bsy{l})R_1(k;\bsy{l})^{-1},
\ee
with
\begin{equation}
\label{25}
\eq{
R_1(k;\bsy{l})&:=
\bma{ccc}
\eins & 0 & 0 \\
0 & \eins & 0\\
0 & 0 & \bsy{u}_+(k;\bsy{l})
\ema,\\ 
R_2(k;\bsy{l})&:=
\bma{ccc}
-\ui k\eins & 0 & 0 \\
0 & \bsy{u}'_-(k;\bsy{0}) & 0\\
0 & 0 & -\bsy{u}'_+(k;\bsy{l})
\ema. 
}
\end{equation}
When the boundary conditions are local in the sense of \eref{localbc} the 
$\mathfrak{S}$-matrix decomposes as
\begin{equation}
\mathfrak{S}(k;P,L)=\bigoplus_{v\in\mathcal{V}}\mathfrak{S}_v(k;P,L).
\end{equation}
We also set
\begin{equation}
\label{wkn}
T(k;\bsy{l}):=\bma{ccc}
0 & 0 & 0\\
0 & 0 & \bsy{u}_-(k;\bsy{l})^{-1} \\
0 & \bsy{u}_+(k;\bsy{l}) & 0
\ema.
\end{equation}
We can now rewrite the expression \eref{neue1} as follows
\begin{equation}
\label{an1}
\eq{
&Z(k;P,L,\bsy{l})\\
  &\hspace{0.5cm}=(P+L)
     \left(\mathds{1}+T(k;\bsy{l})\right)\overline{R_1(\overline{k};\bsy{l})}\\ 
  &\hspace{1.0cm}+P^{\perp}\left(\overline{R_2(\overline{k};\bsy{l})}\,
     \overline{R_1(\overline{k};\bsy{l})}^{-1}+R_2(k;\bsy{l})
     R_1(k;\bsy{l})^{-1}T(k;\bsy{l})\right)\overline{R_1(\overline{k};\bsy{l})}\\
  &\hspace{0.5cm}=\left(P+L+P^{\perp}\overline{D(\overline{k})} 
   + \left( P+L+P^{\perp}D(k) \right)T(k;\bsy{l})\right)
      \overline{R_1(\overline{k};\bsy{l})}\\
  &\hspace{0.5cm}=\left(P+L+P^{\perp}\overline{D(\overline{k})}\right)\left(
    \eins-\mathfrak{S}(k;P,L)T(k;\bsy{l})\right)\overline{R_1(\overline{k};\bsy{l})},
}
\end{equation}
where we used \eref{connection}.

We remark that in contrast to the case of the Laplacian covered in 
\cite{KostrykinSchrader:2006b}, neither of the matrices \eref{71} and \eref{wkn} 
are, in general, unitary when $k\in\rz$. However, choosing a particular system
of admissible fundamental solutions $u^\pm_e(k;\cdot)$ we can construct a matrix $U(k)$ 
that is unitary for $k\in\rz$ such that the positive eigenvalues of $H$ correspond to 
the zeros of $\det(\eins-U(k))$. For this purpose we choose fundamental solutions on 
the internal edges that satisfy
\begin{eqnarray}
\label{inival}
{u_e^+}'(k;0)=\ui k, &\qquad  {u_e^-}'(k;0)=-\ui k,
\end{eqnarray}
and hence $W_e(k)=2\ui k$. Such a pair of fundamental solutions can be easily generated 
from the one in \cite{Pöschl:1987} through a linear combination. Using these fundamental 
solutions we define
\begin{equation}
\label{def:U}
U(k):=R(k)^{-1}\mathfrak{S}(k;P,L)T(k;\bsy{l})R(k),
\end{equation}
with
\be
R(k):=
\bma{ccc}
\eins & 0 & 0\\
0 & \eins & 0\\
0 & 0 & \bsy{r}(k;\bsy{l})
\ema.
\ee
Here $\bsy{r}(k;\bsy{l})$ is an invertible, diagonal matrix with entries $r_e(k;l_e)\neq 0$
on the diagonal, see  Remark~\ref{non_vanishing_of_fund_solutions}.
\begin{lemma}
\label{lem:U}
The matrix $U(k)$ is unitary for $k\in\rz$. 
\end{lemma}
\begin{proof}
Let $k\in\rz$, then
\be
\label{Tnew}
R(k)^{-1}T(k;\bsy{l})R(k)=
\bma{ccc}
0 & 0 & 0\\
0 & 0 & \ue^{\ui\bsy{\phi}(k;\bsy{l})}\\
0 & \ue^{\ui\bsy{\phi}(k;\bsy{l})} & 0
\ema
\ee
is unitary. Here $\ue^{\ui\bsy{\phi}(k;\bsy{l})}$ is a diagonal matrix with diagonal entries
$\ue^{\ui\phi_e(k;\bsy{l})}$, see  Remark~\ref{non_vanishing_of_fund_solutions}. Furthermore, 
it follows from \eref{71} that
\be
\label{Snew}
\eq{
R(k)^{-1}\mfS(k;P,L)R(k)
  &=-\lk (P+L)R(k)+P^{\bot}\overline{D(\overline{k})}R(k)\rk^{-1} \\
  &\qquad\qquad\lk (P+L)R(k)+P^{\bot}D(k)R(k)\rk.
}
\ee
Moreover, 
\be
\label{D(k)}
D(k)=
\bma{ccc}
-\ui k & 0 & 0 \\
0 & -\ui k & 0 \\
0 & 0 & -\bsy{r'}(k;\bsy{l}){\bsy{r}(k;\bsy{l})}^{-1}-\ui k\bsy{r}(k;\bsy{l})^{-2}
\ema,
\ee
where we used the expression for the Wronskian in Remark~\ref{non_vanishing_of_fund_solutions} 
and the fact that $W_e(k)=2\ui k=2\ui{\phi_e}'(k;x)\,r_e(k;x)^2$ for the particular system 
of fundamental solutions we have chosen. This means that 
$R(k)^2\im D(k)=-\ui k\eins$. With
\begin{equation}
\eq{ 
K_1 &:= (P+L)R(k) + P^{\perp}\overline{D(k)}R(k),\\
K_2 &:= (P+L)R(k)+P^{\perp}D(k)R(k),
}
\end{equation}
the fact that $P+L$ and $P^\perp$ are self-adjoint, and $k\in\rz$, the right-hand side of 
\eref{Snew} is $-K_1^{-1}K_2$. This is unitary, if and only if $K_1K_1^\ast=K_2K_2^\ast$. A 
straight-forward calculation confirms that this is indeed the case.
\end{proof}
We are now in a position to characterise the positive eigenvalues of $H$.
\begin{prop}
\label{prop:posev}
Let $k>0$. Then the following statements are equivalent:
\begin{itemize}
\item[(i)] $\det Z(k;P,L,\bsy{l})=0$,
\item[(ii)] $\det(\eins-U(k))=0$,
\item[(iii)] $k^2$ is an eigenvalue of $H$.
\end{itemize}
\end{prop}
\begin{proof}
From Lemma~\ref{14} we know that (iii) implies (i). We now first show the equivalence
of (i) and (ii), using that \eref{an1} implies
\be
\label{ZrelU}
\eq{
&Z(k;P,L,\bsy{l})\\
  &\hspace{0.5cm}=\left((P+L)R(k)+P^{\perp}\overline{D(\overline{k})}R(k)\right)\left(
    \eins-U(k))\right)R(k)^{-1}\overline{R_1(\overline{k};\bsy{l})}.
}
\ee
Thus we need to show that the determinants of 
$(P+L)R(k)+P^{\perp}\overline{D(\overline{k})}R(k)$ and of
$R(k)^{-1}\overline{R_1(\overline{k};\bsy{l})}$ do not vanish when $k>0$. For the second
expression we simply note that both $R(k)$ and $R_1(k;\bsy{l})$ are invertible for 
real $k$. Next, assume that $\det((P+L)R(k) + P^{\perp}\overline{D(\overline{k})}R(k))=0$. 
Then there exists $\bsy{a}\in\kz^E\setminus\{\bsy{0}\}$ such that 
$(P+L+P^{\perp}\overline{D(\overline{k})}){\bsy{a}}=\bsy{0}$ or, equivalently, 
$P\bsy{a}=\bsy{0}$ and $(L+P^{\bot}D(k)P^{\bot})\bsy{a}=\bsy{0}$.
Hence, 
\be
\label{z6}
\left<\bsy{a},\lk L+P^{\bot}D(k)P^{\bot}\rk\bsy{a}\right>_{\kz^E}=
\left<\bsy{a},\lk L+D(k)\rk\bsy{a}\right>_{\kz^E}=0,
\ee
which is equivalent to 
\be
\label{z7}
\left<\bsy{a},L\bsy{a}\right>_{\kz^E} = 
-\left<\bsy{a},\re (D(k))\bsy{a}\right>_{\kz^E}-
\ui\left<\bsy{a},\im(D(k))\bsy{a}\right>_{\kz^E}.
\ee
Since $L$ is self-adjoint, the left-hand side is real, whereas \eref{D(k)} 
implies that the right-hand side has a non-vanishing imaginary part when $k\in\rz$. 
This proves $\det \left((P+L)R(k)+P^{\perp}\overline{D(\overline{k})}R(k)\right)\neq 0$
for $k>0$.

Finally, we have to show that (i) implies (iii). For this assume that $k>0$ is a zero
of $\det Z(k)$. Any solution vector $(\bsy\gamma,\bsy\alpha,\bsy\beta)^T$ from \eref{Z4} 
can be used to construct a function \eref{17}. This is an eigenfunction, iff 
$\bsy{\gamma}=\bsy{0}$. By \eref{ZrelU}, the corresponding solution 
$\bsy{v}=(\bsy{c},\bsy{a},\bsy{b})^T$ of $(\eins-U(k))\bsy{v}=\bsy{0}$ is of the form
\begin{equation}
\bsy{v}=\bma{c}\bsy{c}\\ \bsy{a} \\ \bsy{b}\ema = R(k)^{-1}\overline{R_1(\overline{k};\bsy{l})} 
\bma{c}\bsy\gamma\\ \bsy\alpha \\ \bsy\beta\ema,
\end{equation}
such that $\bsy{\gamma}=\bsy{0}$, iff $\bsy{c}=\bsy{0}$. 

Due to Lemma~\ref{lem:U} the proof of Theorem~3.1 in \cite{KostrykinSchrader:1999} can 
be applied to the current case, leading to $\bsy{c}=\bsy{0}$. Hence there
 exists an eigenfunction corresponding to the eigenvalue $k^2$.
\end{proof}
Altogether, the spectrum of $H$ has the following structure.
\begin{prop}
\label{z8}
The spectrum of $H$ is bounded from below and we have 
\be
\label{z12}
\sigma(H)=
\cases{
\sigma_{pp}(H), & $\mc{E}_{\ex}=\emptyset$,\\
\sigma_{ess}(H)\cup\sigma_{pp}(H), & otherwise,
}
\ee
where $\sigma_{ess}(H)=[0,\infty)$ is the essential spectrum and $\sigma_{pp}(H)$ is the pure 
point spectrum of $H$, respectively. The eigenvalues in $\sigma_{pp}(H)$ have finite 
multiplicities that are bounded by $E$.   
\end{prop}
\begin{proof}
Since the multiplication operator $V$ is bounded in $L^2(\Gamma)$ and $-\Delta$ is 
bounded from below, the semi-boundedness from below of the spectrum of $H$ follows
immediately. 

If $\mc{E}_{\ex}=\emptyset$ the operator $H$ has compact resolvent and hence the spectrum
is pure point and the eigenvalues have finite multiplicities. 

From now on we consider $\mc{E}_{\ex}\neq\emptyset$. The fact that then 
$\sigma_{ess}(H)=[0,\infty)$ follows from noticing that $V$ is bounded and vanishes at 
infinity and, therefore, is relatively compact with respect to $-\Delta$. This is shown 
in complete analogy to \cite[Satz 17.2]{Weidmann:2003b}. Hence 
$\sigma_{ess}(H)=\sigma_{ess}(-\Delta)$, and the latter is well known to be $[0,\infty)$.

The point spectrum is already characterised in Lemma~\ref{14} and in 
Proposition~\ref{prop:posev}.
\end{proof}
In order to characterise the spectrum of $H$ fully the only remaining task is to prove 
the absence of a singularly continuous spectrum. This will follow from an analysis of 
the resolvent and will be given in the next section.
\section{Resolvents}
\label{sec:res}
Our first goal is to identify the resolvent of $H$ as an integral operator.
\begin{defn}[\cite{KostrykinSchrader:2006b}]
\label{defnq}
An operator $K:\mc{D}_K\subset L^2(\Gamma)\to L^2(\Gamma)$ is an integral
operator, if for all $e,e'\in\mc{E}$ there exist functions 
$K_{ee'}(\cdot,\cdot):I_e\times I_{e'}\to\kz$ such that
\begin{enumerate}
\item $K_{ee'}(x,\cdot)\psi_{e'}(\cdot)\in L^1(I_{e'})$ for almost all $x\in I_e$ and all
$\psi=\{\psi_e\}_{e\in\mc{E}}\in\mc{D}_K$,
\item $\phi=K\psi$ with $\psi\in\mc{D}_K$ and
\begin{equation}
\label{19}
\phi_e(x)=\sum_{e'\in\mc{E}}\int_0^{l_{e'}}K_{ee'}(x,y)\,\psi_{e'}(y)\ \ud y.
\end{equation}
\end{enumerate}
\end{defn}
As a shorthand, we sometimes also denote the integral kernel of an integral operator
$K$ as $K(\bsy{x},\bsy{y})=\{K_{ee'}(x_e ,y_{e'})\}_{e,e'\in\mc{E}}$ and its action
\eref{19} as
\begin{equation}
\label{n2}
\phi (\bsy{x}) =\int_\Gamma K(\bsy{x},\bsy{y})\,\psi(\bsy{y})\ \ud\bsy{y}.
\end{equation}
We now show that the resolvent
\begin{equation}
\label{wn1}
R_H (k^2) =\left(H-k^2\right)^{-1}
\end{equation}
of the Schr\"odinger operator $H$ is an integral operator. Here we follow 
\cite{KostrykinSchrader:2006b} closely, where the same was proven for the resolvent of
the Laplacian. 

We require some definitions, the first one being the `free' resolvent kernel
\begin{equation}
\label{24}
r^{(0)}_{ee'}(k;x,y):=\frac{\delta_{ee'}}{W_e(k)}
\cases{
\ue^{ik|x-y|},&  $e\in\mc{E}_{\ex}$,\\
u^+_e(k;x) u^-_e(k;y), & $x\geq y$, $e\in\mc{E}_{\inte}$,\\
u^-_e(k;x) u^+_e(k;y), & $x\leq y$, $e\in\mc{E}_{\inte}$,
}
\end{equation}
where $W_e$ is the Wronskian \eref{def:Wronski} when $e\in\mc{E}_{\inte}$ 
and $W_e(k)=2\ui k$ when $e\in\mc{E}_{\ex}$. We also need the matrix
\begin{equation}
\label{25a}
\Phi(k;\bsy{x}):=
\bma{ccc} 
\ue^{\ui k\bsy{x}} & 0 & 0 \\
0 & \bsy{u}_+(k;\bsy{x}) & \bsy{u}_-(k;\bsy{x})
\ema,
\end{equation}
where $\ue^{\ui k\bsy{x}}$ is a diagonal matrix with diagonal entries $\ue^{\ui kx_e}$,
as well as the diagonal matrix 
\be
\label{24x}
\bsy{W}(k):=
\bma{ccc}
\bsy{W}_{\ex}(k) & 0 & 0\\
0 & \bsy{W}_{\inte}(k) & 0 \\
0 & 0 & \bsy{W}_{\inte}(k) 
\ema,
\ee
where $\bsy{W}_{\ex/\inte}(k)$ are diagonal matrices with the Wronskians $W_e(k)$,
$e\in\mc{E}_{\ex/\inte}$, on the diagonal.
\begin{theorem}
\label{22}
Let $k\in S_\delta\setminus\Sigma_Z$. Then $R_H(k^2)$ is an integral operator with 
integral kernel
\begin{equation} 
\label{56a}
\eq{
&r_{H}(k^2;\bsy{x},\bsy{y})\\
  &\hspace{0.5cm}=r^{(0)}(k;\bsy{x},\bsy{y})+\Phi(k;\bsy{x})
    \overline{R_1(\overline{k},\bsy{l})}^{-1}\left(\mathds{1}-
    \mathfrak{S}(k;P,L)T(k;\bsy{l})\right)^{-1}\\
  &\hspace{1.0cm}\cdot\mathfrak{S}(k;P,L) R_1(k;\bsy{l})\bsy{W}^{-1}(k)\Phi(k;\bsy{y})^T.
}
\end{equation}
\end{theorem}
\begin{proof}
In order to prove that \eref{56a} is the resolvent kernel we first rewrite it as
\begin{equation} 
\label{23}
\eq{
r_{H}\lk k^2;\bsy{x},\bsy{y}\rk
  &=r^{(0)}\lk k;\bsy{x},\bsy{y}\rk-\Phi(k;\bsy{x})Z(k;P,L,\bsy{l})^{-1}\\
  &\hspace{0.5cm}\cdot\lk(P+L)R_1(k;\bsy{l})+P^{\bot}R_2(k;\bsy{l})\rk\bsy{W}^{-1}(k)
     \Phi(k;\bsy{y})^T,
}
\end{equation}
making use of the relation \eref{an1}. We then have to show that for any
$k\in S_\delta\setminus\Sigma_Z$ and for every $\psi\in L^2(\Gamma)$ the function
\begin{equation}
\label{23a}
\phi(\bsy{x}) = \int_\Gamma r_{H}(k^2;\bsy{x},\bsy{y})\,\psi(\bsy{y})\ \ud\bsy{y}
\end{equation}
is in the domain of $H$ and satisfies 
\begin{equation}
\label{23b}
(H-k^2)\phi=\psi.
\end{equation}
We now assume that $k\not\in\Sigma_Z$, so that $Z(k)$ is invertible, and that $k^2$ is not 
in the spectrum $\sigma(H)$ of $H$, hence $R_H(k^2)$ is a bounded operator.  Also,
the explicit form of \eref{23} ensures that the components $\phi_e$ of \eref{23a} are
twice differentiable, hence one can apply $H-k^2$. 

Suppose now that every component $\psi_e$ is continuous on $(0,l_e)$. Direct calculations
for $e\in\mc{E}_{\inte}$ as well as for $e\in\mc{E}_{\ex}$ yield
\begin{equation}
\label{26}
\lk-\frac{\ud^2}{\ud x^2}+V_e(x)-k^2\rk\int_{0}^{l_e}r^{(0)}_{ee}(k;x,y)\,\psi_e(y)\ \ud y
=\psi_e(x).
\end{equation}
Moreover, the matrix entries of $\Phi(k;\bsy{x})$ are eigenfunctions of $H$ (as a formal
differential operator), so that $(H-k^2)\Phi(k;\bsy{x})=0$. This proves \eref{23b} for 
$\psi$ in a dense subset of $L^2(\Gamma)$. As the resolvent is bounded the result can be
extended to $L^2(\Gamma)$. 

In order to prove that \eref{23a} is in the domain of $H$ we first observe that the explicit
form \eref{25} of $r_H^{(0)}$ as well as that of $\Phi$ \eref{24x} imply that 
$\phi\in H^2(\Gamma)$. Hence it remains to verify the vertex conditions. 

We again assume that $\psi_e\in C(0,l_e)$ and find, when $e\in\mc{E}_{\ex}$ and $x$
is close to zero, that
\begin{equation}
\label{27a}
\int_{0}^{l_e}r^{(0)}_{ee}(k;x,y)\psi_e(y)\ \ud y=\frac{\ue^{-\ui kx}}{W_e(k)}
\int_{0}^{l_e}\ue^{\ui ky}\psi_e (y)\ \ud y.
\end{equation}
When $e\in\mc{E}_{\inte}$ and $x$ is close to zero, then
\begin{equation}
\label{27}
\int_{0}^{l_e}r^{(0)}_{ee}(k;x,y)\psi_e(y)\ \ud y=\frac{u^-_e(k;x)}{W_e(k)}\int_{0}^{l_e}
u^+_e(k;y)\psi_e(y)\ \ud y,
\end{equation}
and when $x$ is close to $l_e$, 
\begin{equation}
\label{28}
\int_{0}^{l_e}r^{(0)}_{ee}(k;x,y)\psi_e (y)\ \ud y=\frac{u^+_e(k;x)}{W_e(k)}
\int_{0}^{l_e}u^-_e(k;y)\psi_e (y)\ \ud y.
\end{equation} 
With the abbreviation
\begin{equation}
G(k):=-Z(k;P,L,\bsy{l})^{-1}\lk(P+L)R_1(k;\bsy{l})+P^{\bot}R_2(k;\bsy{l})\rk
\end{equation}
this finally yields
\begin{equation}
\label{29}
\eq{
\underline{\phi}
  &=R_1(k;\bsy{l})\bsy{W}^{-1}(k)\int_{\Gamma}\Phi(k;\bsy{y})^T\psi(\bsy{y})\ \ud\bsy{y}\\
  &\hspace{0.5cm}+X(k;\bsy{l})G(k)\bsy{W}^{-1}(k)\int_{\Gamma}\Phi(k;\bsy{y})^T\psi(\bsy{y})
     \ \ud\bsy{y},\\
\underline{\phi'}
  &=R_2(k;\bsy{l})\bsy{W}^{-1}(k)\int_{\Gamma}\Phi(k;\bsy{y})\psi(\bsy{y})\ \ud\bsy{y}\\
  &\hspace{0.5cm}+Y(k;\bsy{l})G(k)\bsy{W}^{-1}(k)\int_\Gamma\Phi(k;\bsy{y})^T\psi(\bsy{y})
     \ \ud\bsy{y}.
}
\end{equation}
Thus,
\begin{equation}
\label{32}
\eq{
(P+L)\underline{\phi}+P^{\bot}\underline{\phi'}
   &=\lk(P+L){R_1(k;\bsy{l})}+P^{\bot}R_2(k;\bsy{l})\rk\\
   &\qquad\qquad\cdot\bsy{W}^{-1}(k)\int_\Gamma\Phi(k;\bsy{y})^T\psi(\bsy{y})\ \ud\bsy{y}\\
   &\quad+Z(k;P,L,\bsy{l})G(k)\bsy{W}^{-1}(k)\int_\Gamma\Phi(k;\bsy{y})^T\psi(\bsy{y})
       \ \ud \bsy{y}\\
   &=\bsy{0}
}
\end{equation}
This proves the claim for a dense subset of $L^2(\Gamma)$. Since the resolvent is bounded
the result extends to all of $L^2(\Gamma)$.

The right-hand side of \eref{23} is analytic for $k\in S_\delta\setminus\Sigma_Z$ with 
poles in $\Sigma_Z$ due to the zeros of $\det Z(k)$. Hence, the representation \eref{23}
for the resolvent kernel can be extended to $k\in S_\delta\setminus\Sigma_Z$.
\end{proof}
The explicit form \eref{56a} of the resolvent kernel allows us to prove the absence of
a singular continuous spectrum in a way similar to the case of Laplacians on graphs
\cite{Ong:2006}.
\begin{prop}
\label{z8a}
Let $\phi\in C_0^\infty(\Gamma)$. Then the function 
$\langle\phi,\left(R_H(\lambda)-R_H(\lambda)^\ast\right)\phi\rangle$
can be extended from the upper half-plane $\im\lambda>0$ through $\rz_+$ into the
lower half-plane, except for a discrete subset of $\rz_+$. In particular, the singular 
continuous spectrum of $H$ is empty. Hence, $\sigma(H)=\sigma_{ac}(H)\cup\sigma_{pp}(H)$. 
\end{prop}
\begin{proof}
We choose $\phi\in C_0^{\infty}(\Gamma)$ and consider
\be
\label{lt3}
\left<\phi,\im R_H\lk\lambda+\ui\varepsilon\rk\phi\right>,
\quad \lambda>0,\quad\varepsilon\to 0^+,
\ee
where 
\be
\im R_H(\lambda):=\frac{1}{2\ui}\lk R_H(\lambda)-R_H(\lambda)^\ast\rk,
\quad\lambda\in\kz\setminus\sigma(H).
\ee
Representing $\lambda+\ui\varepsilon=k^2$, the limit required in \eref{lt3} can be
achieved by keeping $\re k>0$ fixed and taking $\im k\to 0^+$. We remark that 
$R_H(k^2)^\ast$ is an integral operator whose kernel is the complex conjugate of \eref{56a}.

We now fix two consecutive zeros $0<k_n<k_{n+1}$ of $\det Z(k)$  (corresponding to
consecutive eigenvalues $0<k_n^2<k_{n+1}^2$ of $H$) and choose $0<a<b$ such that
$(a,b)\subset(k_n,k_{n+1})$. The contribution to \eref{lt3} involving the matrix-valued 
integral kernel $r^{(0)}\lk k^2;\cdot,\cdot\rk$ in \eref{56a} is uniformly bounded in 
$\lambda$ taken from a suitable neighbourhood of $(a^2,b^2)$. 

Now choosing the representation \eref{23} for the resolvent kernel one obtains that all
contributions safe of $r^{(0)}\lk k^2;\cdot,\cdot\rk$ depend on $k\in S_\delta$ through the 
fundamental solutions $u_e^\pm$, or $\ue^{ikx}$. Hence, their contribution to \eref{lt3}
is analytic in $k$ except for poles at $k\in\Sigma_Z$ and at $k=0$. Since $\Sigma_Z$ is
discrete with no finite accumulation point, and all positive $k\in\Sigma_Z$ lead to
eigenvalues $k^2$ of $H$, the contribution in question is also uniformly bounded in 
$\lambda$ taken from a suitable neighbourhood of $(a^2,b^2)$. 

Altogether this confirms that there exists a constant $C_{\phi}\geq 0$ such that
\be
\label{z13}
\liminf\limits_{\varepsilon\to 0^+}\sup\limits_{\lambda\in(a,b)}\left<\phi, 
\im R_H\lk \lambda+\ui\varepsilon\rk\phi\right>_{L^{2}(\Gamma)}\leq C_\phi.
\ee  
Hence \cite[Proposition 4.1]{Simon:1987} applies, implying that $H$ has (at most) purely 
absolutely continuous spectrum in $(a,b)$. Since $(a,b)\subset (k_n,k_{n+1})$ can be chosen 
arbitrarily we conclude that the singularly continuous spectrum of $H$ is empty.
\end{proof}
In Proposition~\ref{z8} it was shown that $H$ has a purely discrete spectrum if and only if the
graph is compact. Hence, for non-compact graphs the heat-semigroup $\ue^{-Ht}$, $t>0$, 
is not trace class and, therefore, no heat-trace asymptotics exists. In that case we subtract 
a `free' contribution in such a way that the difference is a trace-class operator. As
\begin{equation}
\label{Laplaceint}
\ue^{-Ht} = \frac{\ui}{2\pi}\int_\gamma\ue^{-\lambda t}\,R_H(\lambda)\ \ud\lambda,
\end{equation}
where $\gamma$ is a contour encircling the spectrum of $H$ with positive orientation, the
relevant difference can be achieved by subtracting a `free' resolvent from $R_H$ in 
\eref{Laplaceint}. This procedure follows \cite{Yafaef:1991,Schrader:2007,Davis:2011}.

The `free' comparison operator is constructed using a graph $\Gamma_{\ex}$ by removing
all internal edges from $\Gamma$ and linking all external edges of $\Gamma$ in a single
vertex. Thus $\Gamma_{\ex}$ is an infinite star graph with $E_{\ex}$ edges. The associated 
Hilbert space is then
\begin{equation}
\label{1ex}
L^2(\Gamma_{\ex})=\bigoplus_{e\in\mc{E}_{\ex}}L^2(0,\infty),
\end{equation}
which is embedded in $L^2(\Gamma)$ in an obvious way, using the embedding operator 
$J_{\ex}:L^2\lk\Gamma_{\ex}\rk\rightarrow L^2(\Gamma)$, 
\be
\label{102}
J_{\ex}(\psi)_{e}:=
\cases{
\psi_e, & $e\in\mc{E}_{\ex}$,\\
0, & $e\in\mc{E}_{\inte}$.
}
\ee
On $\Gamma_{\ex}$ the two comparison operators are the Dirichlet-Laplacian, $-\Delta_D$,
and the Neumann-Laplacian, $-\Delta_N$. Their domains (in $L^2(\Gamma_{\ex})$) are
given in analogy to \eref{4} and \eref{5}, with $P_D=\eins_{E_{\ex}},L_D=0$ and 
$P_N=0,L_N=0$, respectively. 

Both operators, $-\Delta_{D/N}$, are non-negative and self-adjoint. Their resolvents, 
$R_{D/N}\lk k^2\rk:=\lk -\Delta_{D/N}-k^2\rk^{-1}$, and heat-semigroups,
$\exp(t\Delta_{D/N})$, $t>0$, are operators acting on $L^2\lk\Gamma_{\ex}\rk$. They 
can be compared to the resolvent, respectively to the heat-semigroup, of $H$ in terms of 
the operators $J_{\ex}R_{D/N}\lk k^2\rk J_{\ex}^{\ast}$ and 
$J_{\ex}\exp(t\Delta_{D/N})J_{\ex}^{\ast}$ acting
on $L^2(\Gamma)$. By construction, $R_{D/N}$ is the direct sum of the resolvents of 
Dirichlet-, respectively Neumann-, Laplacians on a half-line. Hence, they are integral 
operators with well-known integral kernels from which one immediately obtains the integral 
kernels for $J_{\ex}R_{D/N}\lk k^2\rk J_{\ex}^{\ast}$, as
\be
\label{resolvent_stargraph}
r_{D/N,ee'}\lk k^2;x,y\rk=\delta_{ee'}\frac{\ui}{2k}
\cases{
\ue^{\ui k\left|x-y\right|}\pm\ue^{\ui k\lk x+y\rk}, & $e\in\mc{E}_{\ex}$\\
0, & $e\in\mc{E}_{\inte}$
},
\ee
when $\im k>0$ (see also \cite{Schrader:2007}). 
\begin{prop}
\label{141b}
The difference of resolvents, $R_{H}\lk k^2\rk-J_{\ex}R_{D/N}\lk k^2\rk J_{\ex}^{\ast}$, 
is a trace-class operator and is self-adjoint for $k\in\ui\rz_+\cup\rz_0$. It is an 
integral operator with kernel 
\be
\label{121}
\eq{
 &r_H\lk k^2;\bsy{x},\bsy{y}\rk-r_{D/N}\lk k^2;\bsy{x},\bsy{y}\rk\\
  &\qquad=r^{(0)}\lk k^2;\bsy{x},\bsy{y}\rk-r_{D/N}\lk k^2;\bsy{x},\bsy{y}\rk\\
  &\qquad\quad+\Phi(k;\bsy{x})\overline{R_1(\overline{k},\bsy{l})}^{-1}
         \left(\mathds{1}-\mathfrak{S}(k;P,L)T(k;\bsy{l})\right)^{-1}\\
  &\qquad\qquad\cdot\mathfrak{S}(k;P,L) R_1(k;\bsy{l})\bsy{W}^{-1}(k)\Phi(k;\bsy{y})^T.
}
\ee
The trace of $R_{H}\lk k^2\rk-J_{\ex}R_{D/N}\lk k^2\rk J_{\ex}^{\ast}$ can be expressed as
\be
\label{103}
\eq{
 &\int_{\Gamma}\left[r_H\lk k^2;\bsy{x},\bsy{x}\rk-
    r_{D/N}\lk k^2;\bsy{x},\bsy{x}\rk\right]\ \ud\bsy{x}\\
 &\qquad=\sum_{e\in\mc{E}}\int_0^{l_e}\left[r_{H,ee}\lk k^2;x,x\rk
     -r_{D/N,ee}\lk k^2;x,x\rk\right]\ \ud x.
}
\ee
\end{prop}
\begin{proof}
Self-adjointness is clear since by assumption $k^2\in\rz$ and the operators 
$H$ and $-\Delta_{D/N}$ are self-adjoint. 

In order to prove the trace-class property we introduce the auxiliary graph 
$\Gamma_{\inte}$, obtained by removing the external edges from $\Gamma$. Hence,
$\Gamma_{\inte}$ is a compact graph with edge set $\mc{E}_{\inte}$. We then define the 
Hilbert space $L^2(\Gamma_{\inte})$ and the embedding operator 
$J_{\inte}:L^2\lk\Gamma_{\inte}\rk\rightarrow L^2\lk\Gamma\rk$ in analogy to \eref{1ex} and 
\eref{102}, respectively, interchanging $\mc{E}_{\ex}$ with $\mc{E}_{\inte}$. 
We also require the auxiliary operators $H_{D/N}$ and $H_{\inte,D/N}$, both
acting as the Schr\"odinger operator $H$. The domain of  $H_{D/N}$ consists of
functions in $H^2(\Gamma)$ with Dirichlet/Neumann conditions in the vertices and, 
analogously, the domain of $H_{\inte,D/N}$ comprises of functions in $H^2(\Gamma_{\inte})$ 
with Dirichlet/Neumann conditions. Since $H_{D/N}$ is a finite rank perturbation of $H$ 
we infer that $R_H\lk k^2\rk-R_{H_{D/N}}\lk k^2\rk$ is trace class. Moreover, $H_{\inte,D/N}$ 
acts on a compact graph and can be bounded from above and below (in the sense of
quadratic forms) by $-\Delta_{\inte,D/N}+V_{\min/\max}$, where $V_{\min/\max}$ is the 
minimal/maximal value taken by the potential $V$ on the compact graph $\Gamma_{\inte}$.
The operators $-\Delta_{\inte,D/N}+V_{\min/\max}$ have compact resolvents (see 
\cite{Kuchment:2004}) and their eigenvalue asymptotics follow a Weyl law (in one dimension). 
Hence their resolvents are trace class. This implies that 
$J_{\inte}R_{H_{\inte,D/N}}\lk k^2\rk J_{\inte}^{\ast}$ is also trace class. Moreover, 
by construction $R_{H_{D/N}}\lk k^2\rk-J_{\ex}R_{D/N}\lk k^2\rk J_{\ex}^{\ast}=
J_{\inte}R_{H_{\inte,D/N}}\lk k^2\rk J_{\inte}^{\ast}$. Hence, the difference of resolvents,
\be
\label{5r}
\eq{
R_{H}\lk k^2\rk-J_{\ex}R_{D/N}\lk k^2\rk J_{\ex}^{\ast}
 &=\left[R_{H}\lk k^2\rk-R_{H_{D/N}}\lk k^2\rk\right] \\
 &\hspace{0.5cm}+\left[R_{H_{D/N}}\lk k^2\rk-J_{\ex}R_{D/N}\lk k^2\rk J_{\ex}^{\ast}\right],
}
\ee
is trace class.

By construction, and using Theorem~\ref{22}, the kernel of the
difference $R_{H}\lk k^2\rk-J_{\ex}R_{D/N}\lk k^2\rk J_{\ex}^{\ast}$ is given by 
\eref{121}, where
\begin{equation}
\label{r1t}
\eq{
r^{(0)}_{ee'}\lk k^2;x,y\rk-r_{D/N,ee'}\lk k^2;x,y\rk \\
\qquad = \frac{\delta_{ee'}}{W_e(k)}
\cases{
\mp\ue^{ik(x +y)},&  $e\in\mc{E}_{\ex}$,\\
u^+_e(k;x) u^-_e(k;y), & $x\geq y$, $e\in\mc{E}_{\inte}$,\\
u^-_e(k;x) u^+_e(k;y), & $x\leq y$, $e\in\mc{E}_{\inte}$.
}}
\end{equation}
Since the functions $u^\pm_e(k;x)$ are smooth on every compact interval $I_e$, 
$e\in\mc{E}_{\inte}$, and $\ue^{ik \lk x+y\rk}$ is smooth and square integrable on 
$\rz^+\times \rz^+$ when $\im k>0$, the relation~\eref{103} follows from 
\cite[p. 15]{Schrader:2007} and \cite[p. 117]{Krein:1969}. 
\end{proof}
\section{Asymptotics of the $\mathfrak{S}$-matrix}
\label{sec:asy}
In order to prove heat-kernel asymptotics for small $t$ we employ the relation 
\eref{Laplaceint} and first determine the behaviour of the resolvent kernel for 
large $|k|$. Following Theorem~\ref{22} this requires the asymptotics 
of the $\mf{S}$-matrix. 

For the purpose of asymptotic expansions we now choose a particular systems of
admissible fundamental solutions, see Definition~\ref{def:fundsys}. 
\begin{lemma}[\cite{Fedoryuk:1993,Harrison:2012}]
\label{6}
On each internal edge $e\in\mc{E}_{\inte}$ and for each $k\in S_\delta$ the equation 
\begin{equation}
\label{7}
-u''_e+V_e u_e-k^2 u_e=0
\end{equation}
possesses two linearly independent solutions $u^\pm_e$ such that, for fixed 
$x\in(0,l_e)$, the functions $u^\pm_e(k;x)$ are analytic in $k\in S_{\delta}$, and for 
fixed $k$ they are smooth in $x$. Moreover,  for $|k|\rightarrow\infty$, $k\in S_{\delta}$, 
these solutions possess asymptotic expansions
\begin{equation}
\label{8}
u^\pm_e(k;x) \sim \exp\lk\sum_{l=-1}^{\infty}k^{-l}\int_{0}^{x}\beta_{e,l,\pm}(y)\ 
\ud y\rk,
\end{equation}
that are uniform in $x\in(0,l_e)$. The derivatives ${u^\pm_e}'$ (with respect to $x$)
possess asymptotic expansions in the same domain that are given as the derivatives of 
the right-hand side of \eref{8}.

The coefficient functions $\beta_{e,l,\pm}(x)$ are determined by the recursion relations
\be
\label{101}
\eq{
\beta_{e,l+1,\pm}(x)=\pm\frac{\ui}{2}\left(\beta_{e,l,\pm}'(x)+\sum_{j=0}^l 
\beta_{e,j,\pm}(x)\beta_{e,l-j,\pm}(x)\right)
}
\ee
with $\beta_{e,-1,\pm}(x)=\pm\ui$, $\beta_{e,0,\pm}(x)=0$ and 
$\beta_{e,1,\pm}(x)=\mp\frac{\ui}{2} V_e(x)$.
\end{lemma}
\begin{proof}
The existence of the fundamental solutions possessing the asymptotic behaviour 
\eref{8} is proven in \cite[p. 37,38]{Fedoryuk:1993}. The recursion relations \eref{101} 
for the coefficients are deduced in \cite[p. 12]{Harrison:2012}.
\end{proof}
\begin{rem}
\label{rough_asymptotics}
We note that the leading asymptotic behaviour implied by \eref{101} is
\begin{equation}
\label{uzt}
u^\pm_e(k;x)= \ue^{\pm\ui kx + O(k^{-1})}.
\end{equation}
Similarly \cite{Fedoryuk:1993},
\begin{equation}
(u^\pm_e)'(k;x)= \lk \pm \ui k + O(k^{-1}) \rk u^\pm_e(k;x).
\end{equation}
Hence, asymptotically for large wave numbers, the solutions $u^\pm_e$ are
left- and right-moving, complex, plane waves.
\end{rem}
The recursion relations \eref{101} imply for the next coefficients that
\be
\label{56t}
\eq{
\beta_{e,2,\pm}(x) &=\frac{1}{4}V_e'(x) \\
\beta_{e,3,\pm}(x) &=\pm\frac{\ui}{8}V_e''(x)\mp\frac{\ui}{8}V_e(x)^2 \\
\beta_{e,4,\pm}(x) &=-\frac{1}{16}V_e^{(3)}(x)+\frac{1}{4} V_e'(x) V_e(x).
}
\ee
Using a simple induction on $l$ we also observe that
\be
\label{53}
\beta_{e,l,+}(x)=(-1)^l \beta_{e,l,-}(x)
\ee
and
\be
\label{54}
\beta_{e,l,\pm}(x)\in
\cases{
\rz, & $l$\quad even\\
\ui\rz, & $l$\quad odd.
}
\ee
We note that the condition $u^\pm_e(k;x)=\overline{u^\mp_e(\overline{k};x)}$ required by 
Definition~\ref{def:fundsys} is consistent with equation \eref{53}. 

If one adds further restrictions such as \eref{inival}, an asymptotic expansion as in 
Lemma~\ref{6} need not hold.

The goal of this section is to prove an asymptotic expansion for the $\mathfrak{S}$-matrix
\eref{71} for large $|k|$. We first notice that by setting $V\equiv 0$ the 
$\mathfrak{S}$-matrix \eref{71} becomes the well-known expression
\begin{equation}
\label{SLaplacedef}
\mfS_{-\Delta}(k;P,L)=-\lk P+L+\ui k P^{\bot}\rk^{-1}\lk P+L-\ui k P^{\bot}\rk 
\end{equation}
for the Laplacian (see \cite{KostrykinSchrader:2006b}). 
\begin{defn}
\label{90}
Let 
\begin{equation}
\label{62}
\bsy{\beta}_{j}:=
\bma{ccc}
0 & 0 & 0\\
0 & \beta_{e,j,-}(\bsy{0}) & 0\\ 
0 & 0 & -\beta_{e,j,+}\lk\bsy{l}\rk
\ema.
\end{equation}
We then set
\be
\label{70}
\Lambda_m:=\sum\limits_{j,n\in\nz_0,\atop n+j=m}(\ui L)^nP^{\bot}\overline{\bsy{\beta}_{j+1}}.
\ee
Let $\bsy{m}=(m_1,\dots, m_n) \in{\nz_0}^n$ be a multi-index. Let
\begin{equation}
\label{la1}
\Lambda_{n,r}:=\sum\limits_{|\bsy{m}|=r}\prod_{j=1}^n\Lambda_{m_j},
\end{equation}
where $|\bsy{m}|:=\sum\limits_{j=1}^nm_i$ is the length of the multi-index.
Finally set
\be
\label{om1}
\Omega_j:= \sum\limits_{n \ge 1, \ r \ge 0,\atop r+2n=j}\ui^n\Lambda_{n,r}
\ee
with the convention that $\Omega_0=1$.
\end{defn}
Our main result in this section is the following.
\begin{theorem}
\label{79}
The $\mfS$-matrix for a Schr\"odinger operator admits the asymptotic expansion 
\be
\label{Sasympt}
\mfS\lk k;P,L\rk\sim \mfS_{\infty}+\sum\limits_{m=1}^{\infty}k^{-m}\mf{S}_m,\quad 
|k|\rightarrow\infty,\ k \in S_{\delta},
\ee
where the matrix $\mfS_\infty$ is defined as the large-$k$ limit of the $\mfS$-matrix 
of the Laplacian and is given by
\be
\label{Sequiv}
\mfS_{\infty}:=\eins -2P=\lim_{|k| \rightarrow \infty}\mfS(k,P,L)= 
\lim_{|k| \rightarrow \infty}\mfS_{-\Delta}(k,P,L).
\ee
The perturbative terms are given by
\be
\label{svm}
\mf{S}_m:= \Omega_{m}\mf{S}_{\infty} + 2\sum\limits_{l\geq 0,\ n\geq 1,\atop l+n={m}}
\Omega_l(\ui L)^n+\ui \sum\limits_{r,n,l\in\nz_0,\atop r+n+l=m-2}\Omega_l(\ui L)^r 
P^{\bot}\bsy{\beta}_{n+1} .
\ee
\end{theorem}
In order to prove Theorem \ref{79} we need some auxiliary results. We compute 
the asymptotics of the $\mathfrak{S}$-matrix by comparing it to the case of the 
Laplacian with the help of the following lemma.
\begin{lemma}
\label{75}
The $\mfS$-matrix for the operator $H$ can be written as a perturbation of the 
$\mfS$-matrix of the Laplacian by
\begin{equation}
\label{SwSdelta}
\eq{
\mfS(k;P,L)&=\Omega(k)\mfS_{-\Delta}(k;P,L) \\
           &\quad-\Omega(k)\lk P+L+\ui k P^{\bot}\rk^{-1}P^{\bot} \lk D(k)+ \ui k \rk , 
}
\end{equation}
where the matrix-valued function $\Omega(k)$ is given by
\be
\label{wee1}
\Omega(k):=\lk\eins+\lk P+L+\ui kP^{\bot}\rk^{-1}P^{\bot} \lk \overline{D\lk\overline{k}\rk} 
-\ui k \rk\rk^{-1}.
\ee
\end{lemma}
\begin{proof}
A direct calculation using the definitions \eref{71} and \eref{SLaplacedef} shows that
\begin{equation}
\eq{
 &\Omega(k)\mfS_{-\Delta}(k;P,L) -\Omega(k)\lk P+L+\ui k P^{\bot}\rk^{-1}P^{\bot} \lk D(k)+ 
     \ui k \rk \\
 &\quad =-\Omega(k)\lk P+L+\ui k P^{\bot}\rk^{-1} \left(\lk P+L-\ui k P^{\bot}\rk+ 
           P^{\bot} \lk D(k)+ \ui k \rk\right) \\
 &\quad\quad - \lk\lk P+L+\ui k P^{\bot}\rk+P^{\bot} \lk \overline{D\lk\overline{k}\rk} 
           -\ui k \rk\rk^{-1} \left( P+L+ P^{\bot}  D(k)\right) \\
 &\quad =\mfS(k;P,L).
 }
\end{equation}
\end{proof}
From Remark~\ref{rough_asymptotics} one concludes that $\overline{D\lk\overline{k}\rk}\to\ui k$
when $|k|\to\infty$ in $S_{\delta}$, and thus $\Omega(k)\to\eins$ as well as 
$\mfS(k;P,L)\sim\mfS_{-\Delta}(k;P,L)$, cf.\ \eref{Sequiv}. In order to arrive at an
asymptotic expansion of the $\mf{S}$-matrix more is needed.
\begin{lemma}
\label{77}
The function $\Omega(k)$ possesses an asymptotic expansion 
\be
\Omega(k)\sim  \sum\limits_{l=0}^{\infty}k^{-l}\Omega_l,
\ee
for  $|k|\rightarrow\infty$, $k\in S_{\delta}$.
\end{lemma}
\begin{proof}
We have 
\be
\label{dk1}
D(k)\sim -\ui k+  \sum\limits_{l=1}^{\infty}k^{-l}\bsy{\beta}_{l},\qquad |k|\to\infty
\ee
from the definition of $D(k)$ in equation \eref{56} and of $\bsy{\beta}_l$ in equation 
\eref{62}. This implies $\lk\overline{D\lk\overline{k}\rk} -\ui k \rk = \Or(k^{-1})$, 
so that we can expand $\Omega(k)$ in a power series,
\be
\label{ok2}
\Omega(k)=\sum_{n=0}^{\infty}(-1)^n\lk\lk P+L+\ui k P^{\bot}\rk^{-1} P^{\bot}
\lk \overline{D\lk\overline{k}\rk} -\ui k \rk\rk^n.
\ee
With
\be
\label{64}
\eq{
\lk P+L+\ui k P^{\bot}\rk^{-1}&=P+(\ui k)^{-1}P^{\bot}\lk\eins+(\ui k)^{-1}L\rk^{-1}P^{\bot}\\
&=P-\ui k^{-1}\sum\limits_{r=0}^{\infty}k^{-r}P^{\bot}(\ui L)^{r}P^{\bot}.
}
\ee
Each term in \eref{ok2} can be expanded as $|k|\rightarrow\infty$
(see \cite{Bleistein:1975}),
\be
\label{63}
\eq{
&(-1)^n\lk\lk P+L+\ui k P^{\bot}\rk^{-1} P^{\bot} \lk \overline{D\lk\overline{k}\rk} 
     -\ui k \rk \rk^n\\
&\hspace{0.5cm}\sim\ui^{n} k^{-2n}\lk\sum\limits_{l,r=0}^{\infty}k^{-(l+r)}(\ui L)^rP^{\bot}
     \overline{\bsy{\beta}_{l+1}}\rk^n\\
&\hspace{0.5cm}=\ui^{n}k^{-2n}\lk\sum\limits_{m=0}^{\infty}k^{-m}\Lambda_m\rk^n
               =\ui^{n}\sum\limits_{j=0}^{\infty}k^{-j-2n}\Lambda_{n,j}.
}
\ee 
Hence, as $|k|\rightarrow\infty$,
\be
\label{slm1}
\eq{
&\sum\limits_{n=0}^{\infty}(-1)^n\lk\lk P+L+\ui k P^{\bot}\rk^{-1} P^{\bot} 
    \lk \overline{D\lk\overline{k}\rk} -\ui k \rk \rk^n\\
&\hspace{0.5cm}\sim\eins+\sum\limits_{n=1, \atop j=0}^{\infty}\ui^{n}k^{-j-2n}\Lambda_{n,j}=
    \sum\limits_{l=0}^{\infty}k^{-l}\Omega_l.
}
\ee
\end{proof}
\begin{proof}[Proof of Theorem~\ref{79}]
We finally have all the necessary input to prove Theorem~\ref{79}. We use Lemma~\ref{75} 
and for $\mfS_{-\Delta}(k;P,L)$ we employ a result from \cite{BE:2008},  
\be
\label{hj34}
\mfS_{-\Delta}(k;P,L) \sim \mfS_{\infty} + 2\sum\limits_{n=1}^{\infty}k^{-n}(\ui L)^n,
\quad |k|\rightarrow\infty.
\ee
For the first term on the right-hand side of \eref{SwSdelta} we obtain 
\be
\label{om3}
\Omega(k)\mfS(k;P,L) \sim \sum\limits_{l=0}^{\infty}k^{-l}\Omega_l\mf{S}_{\infty} + 
\sum_{n=1,\atop l=0}^{\infty}k^{-(n+l)}\Omega_l2(\ui L)^n,
\ee
as $|k|\rightarrow\infty$.
Moreover, by Lemma~\ref{77} and equation \eref{64} the second term in \eref{SwSdelta}
gives, 
\be
\label{0m4}
\eq{
&\Omega(k)\lk P+L+\ui k P^{\bot}\rk^{-1}P^{\bot}\lk D(k) + ik\rk \\
&\hspace{0.5cm}\sim-\ui k^{-2}\lk\sum_{l=0}^{\infty}k^{-l}\Omega_l\rk\lk\sum_{r=0}^{\infty}
    k^{-r}(\ui L)^{r}\rk P^{\bot}\lk\sum\limits_{n=0}^{\infty}k^{-n}\bsy{\beta}_{n+1}\rk\\
&\hspace{0.5cm}=-\ui k^{-2} \sum_{l,r,n=0}^{\infty}k^{-l-r-n}\Omega_l (\ui L)^{r}P^{\bot} 
    \bsy{\beta}_{n+1},\quad |k|\rightarrow\infty,
}
\ee
Collecting all the terms finally yields Theorem~\ref{79}.
\end{proof}
\begin{cor}
\label{118}
The first terms of the asymptotic expansion \eref{Sasympt} read
\be
\label{rt56}
\eq{
&\mfS\lk k;P,L\rk\\
&\hspace{0.5cm}= \mfS_{\infty}+2k^{-1}\ui L+2k^{-2}\lk \ui P^{\bot}\bsy{\beta}_1P-L^2\rk\\
&\hspace{1.0cm}+2k^{-3}\lk P^{\bot}\bsy{\beta}_1L-L\bsy{\beta}_1 P+\ui P^{\bot}\bsy{\beta}_2
     P^{\bot}-\ui L^3\rk+\Or\lk k^{-4}\rk,
}
\ee
as $|k|\to\infty$.
\end{cor}
\begin{proof}
The claim follows by Theorem~\ref{79} and Definition~\ref{90}. Using \eref{svm} we find
\begin{equation}
\label{56rtz}
\eq{
 \Omega_2 &= \ui\Lambda_{1,0}= \ui P^{\bot}\overline{\bsy{\beta}_1} \\
 \Omega_3 &= \ui\Lambda_{1,1}= \ui(\ui L)P^{\bot}\overline{\bsy{\beta}_1} + 
              \ui P^{\bot}\overline{\bsy{\beta}_2}.   
} 
\end{equation}
Then, using $\mfS_{\infty} = 1-2P$ and that $\bsy{\beta}_1$ is a purely imaginary-valued 
matrix as well as that $\bsy{\beta}_2$ is a real-valued matrix we get
\begin{equation}
\label{tzu7}
\eq{ 
 \mfS_1 &= 2\ui L \\
 \mfS_2 &= \Omega_2\mfS_{\infty} + 2(\ui L)^2 + \ui P^{\bot} \bsy{\beta}_1 \\
        &= \ui P^{\bot}\overline{\bsy{\beta}_1} -2 \ui P^{\bot}\overline{\bsy{\beta}_1}P 
           - 2L^2 + \ui P^{\bot} \bsy{\beta}_1 \\
 \mfS_3 &= \Omega_3\mfS_{\infty} + 2(\ui L)^3 + \Omega_2 2(\ui L) + \ui (\ui L)P^{\bot} 
           \bsy{\beta}_1+ \ui P^{\bot} \bsy{\beta}_2 \\
        &= -LP^{\bot}\overline{\bsy{\beta}_1} + 2LP^{\bot}\overline{\bsy{\beta}_1}P +
           \ui P^{\bot}\overline{\bsy{\beta}_2} -2\ui P^{\bot}\overline{\bsy{\beta}_2}P \\
        &\hspace{1cm}- 2\ui L^3 - 2P^{\bot}\overline{\bsy{\beta}_1} L  -L P^{\bot} \bsy{\beta}_1
            + \ui P^{\bot} \bsy{\beta}_2 .
}
\end{equation}
\end{proof}
We will also need the asymptotic behaviour of
$\mathfrak{S}(k;P,L)T(k;\bsy{l})$.
\begin{cor}
\label{105}
For $|k| \rightarrow\infty$ with $k\in S_{\delta}$,
\be
\label{58}
\mathfrak{S}(k;P,L)T(k;\bsy{l}) = \mf{S}_\infty T_\infty(k) + \Or\lk k^{-1}\rk,
\ee
where
\be
T_{\infty}(k;\bsy{l}):= 
\bma{ccc}
0 & 0 & 0\\
0 & 0 & e^{\ui k\bsy{l}}\\ 
0 & e^{\ui k\bsy{l}} & 0
\ema.
\ee
\end{cor}
\begin{proof}
The statement follows immediately from Theorem~\ref{79} together with \eref{uzt} and the 
definition \eref{wkn} of $T(k)$.
\end{proof}
\section{Asymptotics of the Wronskian and related terms}
\label{asyml}
The representation \eref{56a} of the resolvent kernel requires the knowledge of some further
quantities, among them the Wronskians \eref{def:Wronski} associated with internal edges
$e\in\mc{E}_{\inte}$. To compute the asymptotics of the inverse of the Wronskian we need 
a similar definition to Definition~\ref{90}.
\begin{defn}
\label{defk}
\begin{itemize}
\item[(i)] Let $n\in\nz$ and let $\bsy{m} \in \nz^n_0$ be a multi-index. Then we set  
\be
\beta^{\bsy{m}}_{e}(x):=\ui^n\prod\limits_{j=1}^n\beta_{e,2m_j+1,+}(x).
\ee
\item[(ii)] We also define the coefficients 
\be
\label{wlk}
w_{e,l}(x):=
\cases{
\sum\limits_{n=1}^{l} \sum\limits_{|\bsy{m}|=l-n}\beta^{\bsy{m}}_{e}(x), & $l\in\nz$,\\
1, & $l=0$.
}
\ee
\end{itemize}
\end{defn}
\begin{lemma}
\label{Wronski_trick}
The Wronskian associated with an internal edge has the following asymptotic expansion,
\begin{equation}
\label{wek1}
W_e(k) \sim-2u^+_e(k;x)u^-_e(k;x) \sum_{l=-1}^{\infty}k^{-(2l+1)}\beta_{e,2l+1,+}(x), 
\end{equation}
when $|k|\to\infty$ with $k \in S_{\delta}$.
\end{lemma}
\begin{proof}
We have
\begin{equation}
\label{stz56}
\eq{
W_e(k) &=u^+_e(k;x){u^-_e}'(k;x) - u^-_e(k;x){u^+_e}'(k;x)   \\
       &\sim u^+_e(k;x)\lk \sum_{l=-1}^{\infty}k^{-l}\beta_{e,l,-}(x)\rk u^-_e(k;x)\\
       &\hspace{1.0cm} - u^-_e(k;x)\lk \sum_{l=-1}^{\infty}k^{-l}\beta_{e,l,+}(x)\rk 
           u^+_e(k;x) \\
       &=u^+_e(k;x)u^-_e(k;x)\sum_{l=-1}^{\infty}k^{-l}\lk \beta_{e,l,-}(x)-\beta_{e,l,+}(x)\rk\\
       &=-2u^+_e(k;x)u^-_e(k;x) \sum\limits_{l=-1}^{\infty}k^{-(2l+1)}\beta_{e,2l+1,+}(x) 
},
\end{equation}
where we used the symmetry relations of the coefficients given in \eref{53} and \eref{54}.
\end{proof}
This result is useful to obtain another asymptotic expansion needed in the resolvent.
\begin{lemma}
\label{Wronski_cancellation}
As $|k|\to\infty$ with $k\in S_{\delta}$ the following asymptotic expansion holds,
\begin{equation}
\label{we3}
\eq{
\frac{1}{W_e(k)} u^+_e(k;x)u^-_e(k;x)\sim
-\frac{1}{2\ui k} \sum_{l=0}^\infty k^{-2l}w_{e,l}(x),
}
\end{equation}
where $w_{e,l}$ is defined in \eref{wlk}.
\end{lemma}
\begin{proof}
This follows from a direct application of Lemma~\ref{Wronski_trick}.
\begin{equation}
\label{Wronski_to_wl}
\eq{
& \frac{1}{W_e(k)} u^+_e(k;x)u^-_e(k;x)\\
&\hspace{0.5cm}\sim -\lk\sum\limits_{l=-1}^{\infty}2k^{-(2l+1)}\beta_{e,2l+1,+}(x)\rk^{-1}\\
&\hspace{0.5cm}=-\lk 2\ui k\rk^{-1}\lk 1-\lk\sum\limits_{l=0}^{\infty}\ui k^{-2l-2}
    \beta_{e,2l+1,+}(x)\rk\rk^{-1}\\
&\hspace{0.5cm}=-\lk 2\ui k\rk^{-1}\sum\limits_{n=0}^{\infty}\lk\sum\limits_{l=0}^{\infty}
    \ui k^{-2l-2}\beta_{e,2l+1,+}(x)\rk^n\\
&\hspace{0.5cm}=-\lk 2\ui k\rk^{-1}\sum\limits_{l=0}^\infty k^{-2l}w_{e,l}(x).
}
\end{equation}
\end{proof}
The leading terms of this expansion can be worked out explicitly. Using
\be
\label{wek3}
w_{e,1}(x)=\beta^0_{e}(x)=\ui\beta_{e,1,+}(x) = \frac{1}{2}V_{e}(x)
\ee
and
\be
\label{w2k1}
\eq{
w_{e,2}(x)&=\beta_{e}^{1}(x)+\beta_{e}^{(0,0)}(x)=\ui\beta_{e,3,+}(x)-{\beta_{e,1,+}}(x)^2 \\
         &=-\frac{1}{8}V''_{e}(x) + \frac{3}{8} V^2_{e}(x),
}
\ee
we find that
\be
\label{hjk3}
\eq{
&\frac{1}{W_e(k)}u^+_e(k;x)u^-_e(k;x) \\
&\qquad\sim -\frac{1}{2\ui}k^{-1}-\frac{V_{e}(x)}{4\ui }k^{-3}
+\frac{ V''_{e}(x) - 3 V^2_{e}(x)}{16\ui}k^{-5}+\Or\lk k^{-7}\rk.
}
\ee
We also need the asymptotics of the following two expressions in the upper half-plane.
For that purpose we introduce the sector 
\begin{equation}
\label{S_delta+}
S_{\delta}^+ := \left\{ z\in\mathbb{C};\ 0 < |z| < \infty,\ 
\left|\arg(z)-\frac{\pi}{2}\right| < \frac{\pi}{2}-\delta \right\}\subset S_\delta
\end{equation}
for $\delta>0$ (supposed to be small). In particular, $k\in S^+_\delta$ implies $\im k>0$.
\begin{lemma}
\label{122}
Let $e\in\mc{E}_{\inte}$ and let $k$ be confined to the sector $S_\delta^+$. Then
the two expressions below possess a complete asymptotic expansion as $|k|\to\infty$,
whose leading terms are
\be
\label{wk5}
\eq{
& \frac{1}{W_e(k)}\int_0^{l_e}u^+_e(k;x)^2\ \ud x \\
&\hspace{1cm}=-\frac{1}{4k^2} - \frac{1}{4 k^4}V_e(0) + \frac{1}{8 \ui k^5}V_e(0) +
   \Or\lk k^{-6}\rk,
}
\ee
and
\be
\label{wgn}
\eq{
&\frac{1}{W_e(k)}\frac{u^+_e\lk k;l_e\rk}{u^-_e\lk k;l_e\rk}\int_0^{l_e}u^-_e(k;x)^2\ \ud x \\
&\hspace{1cm}=- \frac{1}{4k^2} - \frac{1}{4k^4}V_e(l_e) - \frac{1}{8\ui k^5}V'_e(l_e) 
      + \Or(k^{-6}).
}
\ee
\end{lemma}
\begin{proof}
In view of the asymptotic behaviour \eref{uzt} of the fundamental system we set
\be
u^\pm_e(k;x)=\ue^{\pm\ui kx}\,v^\pm_e(k;x)
\ee
and notice that for fixed $x$ the functions $v^\pm_e(k;x)$ are bounded in $k$, and 
their derivatives are of the order $\Or(k^{-1})$.

Integrating by parts $(N+1)$-times and using that $\im k>0$ when $k\in S^+_\delta$ we find 
for \eref{wk5} that
\be
\label{112}
\eq{
&\int_{0}^{l_e}\ue^{2\ui kx}v^+_e(k;x)^2\ \ud x \\
&\qquad=\sum\limits_{n=0}^{N}\frac{1}{(-2\ui k)^{n+1}}\frac{\ud^n}{\ud x^n}
         \left.\lk v^+_e(k;x)^2\rk\right|_{x=0}+\Or\lk k^{-N-2}\rk.
}
\ee
From Lemma~\ref{6} we find that
\be
\label{111}
\eq{
v^+_e(k;0)^2 &=1\\
\frac{\ud}{\ud x}\left.\lk v^+_e(k;x)^2\rk\right|_{x=0}
   &=2k^{-1}\beta_{e,1,+}(0) + 2k^{-2}\beta_{e,2,+}(0) + \Or\lk k^{-3}\rk\\
\frac{\ud^2}{\ud x^2}\left.\lk v^+_e(k;x)^2\rk\right|_{x=0}
  &=2 k^{-1}\beta_{e,1,+}'(0)+\Or\lk k^{-2}\rk\\
\frac{\ud^3}{\ud x^3}\left.\lk v^+_e(k;x)^2\rk\right|_{x=0}
  &=\Or\lk k^{-1}\rk,
}
\ee
and using this in \eref{112} gives
\be
\label{eng2}
\eq{
&\int_{0}^{l_e}u^+_e(k;x)^2\ \ud x \\
&\hspace{0.5cm}=-\frac{1}{2\ui k}  - \frac{\beta_{e,1,+}(0)}{2k^{3}}
-\frac{1}{k^4}\lk\frac{\beta_{e,2,+}(0)}{2}+\frac{\ui\beta_{e,1,+}'(0)}{4}\rk
+\Or\lk k^{-5}\rk.
}
\ee
An expansion for ${1}/{W_e(k)}$ follows from Lemma~\ref{Wronski_cancellation} and 
\eref{hjk3} evaluated at $x=0$. Altogether this yields 
\begin{equation}
\label{wek6}
\eq{
&\frac{1}{W_e(k)}\int_0^{l_e}u^+_e(k;x)^2\ \ud x  \\
&\hspace{0.5cm}=\lk-\frac{1}{2\ui k} - \frac{V_{e}(0)}{4\ui k^3} + \Or\lk k^{-5}\rk \rk 
     \lk-\frac{1}{2\ui k}-\frac{V_e(0)}{4\ui k^3} - \frac{V_e'(0)}{4 k^4}+\Or\lk k^{-5}\rk \rk \\
&\hspace{0.5cm}=-\frac{1}{4k^2} - \frac{V_e(0)}{4 k^4} + \frac{V_e'(0)}{8 \ui k^5} 
    +\Or\lk k^{-6}\rk.
}
\end{equation}
For \eref{wgn} we use Lemma~\ref{Wronski_trick} evaluated at $x=l$, and then integrate by
parts $(N+1)$-times as in \eref{112},
\be
\label{wek7}
\eq{
&\frac{1}{W_e(k)}\frac{u^+_e\lk k;l_{e}\rk}{u^-_e\lk k;l_{e}\rk} 
     \int_0^{l_e}u^-_e\lk k;x \rk^2\ \ud x \\
&\hspace{0.5cm}\sim\frac{1}{ -2 \sum\limits_{l=-1}^{\infty}k^{-(2l+1)}\beta_{e,2l+1,+}(l_e)}
     \int_{0}^{l_e}\ue^{-2\ui k(x-l_e)}\,\frac{v^-_e\lk k;x \rk^2}{v^-_e\lk k;l_{e}\rk^2}\ \ud x\\
&\hspace{0.5cm}\sim \frac{1}{2ik}\lk \sum_{l=0}^{\infty}k^{-2l}w_{e,l}(l_e) \rk\\
&\hspace{1.0cm}\cdot\lk\sum_{n=0}^N\frac{1}{(2\ui k)^{n+1}}\frac{\ud^n}{\ud x^n}
     \left.\lk v^-_e(k;x)^2\rk\right|_{x=0}+\Or\lk k^{-N-2}\rk\rk
}
\ee
The prefactor was computed as in \eref{Wronski_to_wl}. We still need
\begin{equation}
\label{dek3}
\eq{
v^-_e(k;0) &= 1 \\
\frac{\ud}{\ud x}\left.\lk v^-_e(k;x)^2\rk\right|_{x=0}
   &= 2k^{-1}\beta_{e,1,-}(l_e) + 2k^{-2}\beta_{e,2,-}(l_e)+\Or(k^{-3})\\
\frac{\ud^2}{\ud x^2}\left.\lk v^-_e(k;x)^2\rk\right|_{x=0}
   &=2k^{-1}\beta_{e,1,-}'(l_e)+ \Or(k^{-2})\\
\frac{\ud^3}{\ud x^3}\left.\lk v^-_e(k;x)^2\rk\right|_{x=0}
   &= \Or(k^{-1}),
}
\end{equation}
giving
\begin{equation}
\label{er6}
\eq{
&\sum_{n=0}^{N}\frac{1}{(2\ui k)^{n+1}}\frac{\ud^n}{\ud x^n}
         \left.\lk v^-_e(k;x)^2\rk\right|_{x=0}+\Or\lk k^{-N-2}\rk \\
&\hspace{0.5cm}=\frac{1}{2ik}-\frac{1}{2k^3}\beta_{e,1,-}(l_e) - 
         \frac{1}{4k^4}\lk 2\beta_{e,2,-}(l_e) -i\beta_{e,1,-}'(l_e) \rk +\Or(k^{-5}) \\
}
\end{equation}
Now using $w_{e,0}(l_e)=1$ and \eref{wek3} finally yields \eref{wgn}.
\end{proof}
\section{Asymptotics of the resolvent kernel}
\label{l2}
The results of the previous sections enable us to determine the precise asymptotic behaviour of 
the resolvent kernel for $k$ in the sector $S^+_\delta$. As expected, the result on the 
diagonal is different from that off the diagonal, and this distinction carries over to the heat
kernel.
\begin{lemma}
\label{lem:resoffd}
In the sector $S_{\delta}^+$ the resolvent kernel of $H$ satisfies the estimate
\be
\label{local1}
r_{H,ee'}(k^2;x,y) = \Or\lk k^{-1}\ue^{-\im k d(x,y)}\rk, 
\ee
when $|k|\rightarrow \infty$. 
\end{lemma}
\begin{proof}
For the proof we use the representation \eref{56a} of the resolvent kernel. The contribution
of the `free' part $r^{(0)}_{ee'}(k;x,y)$ is clearly of the form \eref{local1}, see \eref{24}.

For the remaining contributions we note that as $|k|\to\infty$ with $k\in S_\delta$,
\begin{equation}
\label{respath1}
\eq{
\Phi(k;\bsy{x})\overline{R_1(\overline{k},\bsy{l})}^{-1} 
&= \bma{ccc}\ue^{ik\bsy{x}} & 0 & 0 \\ 0 & \bsy{u}_+(k;\bsy{x}) & 
      \bsy{u}_-(k;\bsy{x})\bsy{u}_-(k;\bsy{l})^{-1}\ema \\
&\sim \bma{ccc}\ue^{ik\bsy{x}} & 0 & 0 \\ 0 & \ue^{\ui k\bsy{x}} & 
      \ue^{\ui k(\bsy{l}-\bsy{x})}\ema,
}
\end{equation}
as well as
\begin{equation}
\label{respath2}
\eq{
 R_1(k;\bsy{l})\bsy{W}^{-1}(k)\Phi(k;\bsy{y})^T 
 &= \bsy{W}^{-1}(k)\bma{cc}\ue^{ik\bsy{y}} & 0 \\ 0 & \bsy{u}_+(k;\bsy{y}) \\ 
       0 & \bsy{u}_+(k;\bsy{l})\bsy{u}_-(k;\bsy{y}) \ema \\
 &\sim \frac{\ui}{2k}\bma{cc}\ue^{ik\bsy{y}} & 0 \\ 0 & \ue^{\ui k\bsy{y}} \\ 
       0 & \ue^{\ui k(\bsy{l}-\bsy{y})} \ema .
}
\end{equation}
The remaining expression $\lk\eins-\mfS(k;P,L)T(k;\bsy{l})\rk^{-1}\mfS(k;P,L)$ has an
asymptotic behaviour that follows from Theorem~\ref{79} and Corollary~\ref{105}. The
latter implies, in particular, that for sufficiently large $|k|$ and $\im k>0$ the matrix 
$\mfS(k;P,L)T(k;\bsy{l})$ has norm less than one. Moreover, for such values of $k$,
\be
\label{geometricST}
\eq{
&\lk\eins-\mfS(k;P,L)T(k;\bsy{l})\rk^{-1}\mfS(k;P,L)\\
&\hspace{1cm}=\sum_{n=0}^{\infty}\lk\mfS(k;P,L)T(k;\bsy{l})\rk^n\mfS(k;P,L)\\
&\hspace{1cm}=\sum_{n=0}^{\infty}\lk\mfS_\infty T_\infty(k;\bsy{l})\rk^n\mfS_\infty
   +\Or\lk k^{-1}\rk .
}
\ee
Putting \eref{respath1}-\eref{geometricST} together according to \eref{56a}  we notice that,
asymptotically, the same expression emerges for the `non-free' contribution to the resolvent
kernel as in the case where $H$ is a Laplacian, see \cite[Proposition 3.3.]{Schrader:2007}.
Hence, asymptotically this contribution can be represented as a sum over paths from $y$
to $x$ on $\Gamma$. Let $\mc{P}_{xy}$ the set of such paths, and let $d_p(x,y)$ be the
distance of $x$ and $y$ along $p_{xy}\in\mc{P}_{xy}$, then 
\be
\label{pathsxeyf}
\eq{
r_{H,ee'}\lk k^2;x,y\rk 
  &\sim \delta_{ee'}\frac{\ui}{2k}\,\ue^{\ui k|x-y|} +\Or\lk k^{-2}\rk \\
  &\quad +\frac{\ui}{2k}\sum_{p_{xy}\in\mc{P}_{xy}} \lk A_p+\Or\lk k^{-1}\rk\rk\,\ue^{\ui k d_p(x,y)},
}
\ee
when $|k|\to\infty$ with $k\in S^+_\delta$. Here the amplitudes $A_p$ arise from
multiplying matrix elements of $\mf{S}_\infty$ along the path $p_{xy}$ in the same way as
in \cite{Schrader:2007}. As the distance $d(x,y)$ is the minimum of $d_p(x,y)$ over all
$p_{xy}\in\mc{P}_{xy}$, the result \eref{local1} follows.
\end{proof}
%
%
%
%
%
%
%
%
%
%
According to Lemma~\ref{lem:resoffd} points $x,y$ on the graph with zero distance deserve
a further investigation. When such points $x,y$ are at edge ends, particular care has to be 
taken as to which edge the points belong to. We clarify this case by using the notation
$x \cong (v,e)\in\mc{V}\times\mc{E}$ indicating that the point is an edge end of $e$ and $v$ is the 
vertex adjacent to $e$. This notation is well defined since we have excluded tadpoles in our 
construction.  Moreover, 
in the following we use for the corresponding matrix elements the notation $M_{xy}:=M_{ee'}$, 
$x\cong(v,e)$, $y \cong (v,e')$, where $v$ has to be adjacent to $e$ and $e'$. This notation is well-defined 
since we have assumed local boundary conditions and no tadpoles.
\begin{lemma}
\label{l1}
Let $x,y$ be points on $\Gamma$ with $d\lk x,y\rk=0$. 
Then, for $k\in S^+_\delta$ the resolvent
kernel of $H$ possesses a complete asymptotic expansion as $|k|\to\infty$ in powers of
$k^{-1}$. The leading terms are given by:
\begin{itemize}
\item[(i)] $x$ is not an edge end:
\be
\label{l7b}
r_{H,ee}\lk k; x,x\rk_{ee}=\frac{\ui}{2k}+\Or(k^{-3}),
\ee
\item[(ii)] $x \cong (v,e)$ and $y \cong (v,e')$ with $e \neq e'$:
\be
\label{l7a}
r_{H,ee'}\lk k; x,y\rk=\frac{\ui{\mfS_{\infty,xy}}}{2k}-\frac{L_{xy}}{k^2}+\Or\lk k^{-3}\rk,
\ee
\item[(iii)] $x \cong (v,e)$: 
\be
\label{l6}
r_{H,ee}\lk k; x,x\rk = \frac{\ui}{2k}\lk 1+{\mfS_{\infty,xx}}\rk-\frac{L_{xx}}{k^2}
+\Or\lk k^{-3}\rk.
\ee
\end{itemize}
\end{lemma}
\begin{proof}
\label{l7}
We use the expansion \eref{geometricST} in the representation \eref{121} 
of the resolvent kernel. The coefficients of $k^{-1}$ and $k^{-2}$ in \eref{l7b}, \eref{l7a} 
and \eref{l6} follow from the leading term, coming from $n=0$, in \eref{geometricST}.
For $e,e'\in\mc{E}_{\inte}$ we have
\be
\label{matrixm}
\eq{
&\lk\Phi(k;\bsy{x})\overline{R_1(\overline{k},\bsy{l})}^{-1}\mfS\lk k;P,L\rk R_1(k;\bsy{l})
      \bsy{W}^{-1}(k)\Phi(k;\bsy{y})^T\rk_{ee'}\\
&\hspace{1cm}=u^+_e(k;x_e)\mf{S}\lk k;P,L\rk_{ee'}{W_{e'}}^{-1}(k)u^+_{e'}(k;y_{e'})\\
&\hspace{1.5cm}+u^+_e(k; x_e)\mf{S}\lk k;P,L\rk_{ee'}u^+_{e'}(k;l_{e'})W_{e'}^{-1}(k)
      u^-_{e'}(k;y_{e'})\\
&\hspace{1.5cm}+u^-_e(k;x_e)u^-_e(k;l_e)^{-1} \mfS(k;P,L)_{ee'}W_{e'}^{-1}(k)u^+_{e'}(k;y_{e'})\\
&\hspace{1.5cm}+u^-_e(k;x_e)u^-_e(k;l_e)^{-1} \mfS(k;P,L)_{ee'}u^+_{e'}(k;l_{e'})W_{e'}^{-1}(k) 
      u^-_{e'}(k;y_{e'}).
}
\ee
Analogous results hold for the other three cases. A straightforward incorporation of the 
leading two coefficients with respect to $k$ coming from \eref{8}, \eref{rt56} and \eref{wek1} 
in \eref{matrixm} completes the proof.
\end{proof}
We now have all the ingredients to determine the trace of a regularised resolvent,
in the sense of Proposition~\ref{141b}.
\begin{theorem}
\label{151}
Let $\Gamma$ be a compact or non-compact metric graph. Then the trace of the difference of 
resolvents, $R_H(k^2)-J_{\ex}R_{D/N}(k^2)J^\ast_{\ex}$, possesses an asymptotic expansion in 
powers of $k^{-1}$ when $|k|\rightarrow\infty$ with $k$ in the sector $S^+_\delta$ of the 
upper half-plane,
\be
\label{150a}
\tr\lk R_H(k^2)-J_{\ex}R_{D/N}(k^2)J^\ast_{\ex}\rk \sim \sum_{n=1}^{\infty}b_nk^{-n}.
\ee
The leading coefficients are given by 
\begin{equation}
\label{150}
\eq{
 b_1 &= -\frac{\mathcal{L}}{2\ui } \\
 b_2 &= -\frac{1}{4}\tr\mfS_{\infty}\pm\frac{E_{\ex}}{4}\\
 b_3 &= - \frac{1}{4\ui} \int_{\Gamma_{\inte}} V(\bsy{x})\ \ud\bsy{x} + \frac{1}{2\ui}\tr L\\
 b_4 &= -\frac{1}{4}\sum_{v \in \mc{V}} \sum\limits_{e \sim v} \mf{S}_{\infty,ee}V_e(v) + 
               \frac{1}{2}\tr L^2 \\
 b_5& = \frac{1}{16\ui}\int_{\Gamma_{\inte}} \lk V''(\bsy{x}) - 3 V^2(\bsy{x})\rk \ud\bsy{x} + 
                \frac{1}{8\ui}\sum_{v \in \mc{V}} \sum_{e \sim v} \mf{S}_{\infty,ee}V_{e}'(v) \\
       &\hspace{0.5cm} + \frac{3}{4\ui} \sum_{v \in \mc{V}}\sum_{e \sim v} L_{ee} V_e(v) - 
               \frac{1}{2 \ui}\tr L^3
 - \frac{1}{8\ui}\sum\limits_{v\in\mc{V}}\sum\limits_{e\sim v}P^{\bot}_{ee}V_e'(v).
}
\end{equation}
\end{theorem}
\begin{proof}
By Proposition~\ref{141b} the difference of resolvents is trace class and the trace can
be obtained from the resolvent kernel as in \eref{103} with the kernel \eref{121}.

The first contribution comes from the `free' part \eref{r1t},
\begin{equation}
\label{iop}
\eq{
&\int_{\Gamma} \lk r^{(0)}\lk k^2;\bsy{x},\bsy{x}\rk-r_{D/N}\lk k^2;\bsy{x},\bsy{x}\rk \rk
        \ud\bsy{x} \\
&\hspace{0.5cm}=\mp\sum_{e\in\mc{E}_{\ex}}\frac{\ui}{2k}\int_{0}^{\infty}\ue^{2\ui kx}\ \ud x
        +\int_{\Gamma_{\inte}}\frac{1}{\bsy{W}_{\inte}(k)}\bsy{u}_+(k;\bsy{x})
        \bsy{u}_-(k;\bsy{x})\ \ud\bsy{x}.
}
\end{equation}
For the first term on the right-hand side we use that $\im k>0$. The asymptotics of the 
second term follow from Lemma~\ref{Wronski_cancellation} and \eref{hjk3},
\begin{equation}
\label{iop1}
\eq{
&\int_{\Gamma} \lk r^{(0)}\lk k^2;\bsy{x},\bsy{x}\rk-r_{D/N}\lk k^2;\bsy{x},\bsy{x}\rk \rk
        \ud\bsy{x} \\
&\hspace{1cm}\sim \pm\frac{E_{\ex}}{4k^2}-\frac{\mathcal{L}}{2\ui k} - \frac{1}{4\ui k^3} 
        \int_{\Gamma_{\inte}}V(\bsy{x})\ \ud\bsy{x}\\
&\hspace{2cm}+\frac{1}{16\ui k^5}\int_{\Gamma_{\inte}}\lk V''(\bsy{x}) - 3V^2(\bsy{x})\rk
         \ud\bsy{x} + \Or\lk k^{-7}\rk.
}
\end{equation}
The `non-free' contribution, i.e., the second and third lines on the right-hand side of 
\eref{121}, will give the contribution at the vertices. We first observe that
Corollary~\ref{105} implies, when $\im k>0$, that 
$\lk\eins-\mf{S}(k;P,L)T(k;\bsy{l})\rk)^{-1}=\eins+\Or(k^{-\infty})$. After a cyclic 
permutation we hence have to determine the asymptotic expansion of
\be
\label{106}
\tr\lk\int_{\Gamma}\mfS\lk k;P,L\rk R_1(k;\bsy{l})\bsy{W}(k)^{-1}\phi(k;\bsy{x})^T\phi(k;\bsy{x})
\overline{R_1\lk\overline{k},\bsy{l}\rk}^{-1} \ud \bsy{x} \rk.
\ee
As the $\mf{S}$-matrix is independent of $\bsy{x}$ we only need to integrate the remaining
terms. For these, a straight-forward calculation gives
\begin{equation}
\label{tzn}
\eq{
&R_1(k;\bsy{l})\bsy{W}(k)^{-1} \phi(k;\bsy{x})^T\phi(k;\bsy{x})
       \overline{R_1\lk\overline{k},\bsy{l}\rk}^{-1}\\
&\hspace{2cm}=
\bma{ccc}
\ue^{2\ui k\bsy{x}} & 0 & 0\\
0 & \frac{\bsy{u}_+(k;\bsy{x})^2}{\bsy{W}_{\inte}(k)} 
& \frac{\bsy{u}_+(k;\bsy{x})\bsy{u}_-(k;\bsy{x})\bsy{u}_+(k;\bsy{l})}{\bsy{W}_{\inte}(k)}\\
0 & \frac{\bsy{u}_-(k;\bsy{x})\bsy{u}_+(k;\bsy{x})}{\bsy{W}_{\inte}(k)\bsy{u}_-(k;\bsy{l})} 
& \frac{\bsy{u}_-(k;\bsy{x})^2\bsy{u}_+(k;\bsy{l})}{\bsy{W}_{\inte}(k)\bsy{u}_-(k;\bsy{l})}  
\ema .
}
\end{equation}
%
%
We note that $\ue^{2\ui k\bsy{x}}=\Or(k^{-\infty})$, and that asymptotic expansions of the 
integrals of the remaining diagonal blocks were determined in Lemma~\ref{122}. 
Lemma~\ref{Wronski_cancellation} implies that 
$\bsy{W}_{\inte}(k)^{-1}\bsy{u}_\pm(k;\bsy{x})\bsy{u}_\mp(k;\bsy{x})$ possess asymptotic 
expansions in powers of $k^{-1}$. In the off-diagonal blocks of \eref{tzn}, however, these 
terms are multiplied by $\bsy{u}_\pm(k;\bsy{l})^{\pm 1}=\Or(k^{-\infty})$. Therefore, the
off-diagonal blocks do not contribute to the expansion. Thus,
\begin{equation}
\label{boundary_terms}
\eq{
&\int_{\Gamma} R_1(k;\bsy{l})\bsy{W}(k)^{-1} \phi(k;\bsy{x})^T\phi(k;\bsy{x})
       \overline{R_1\lk\overline{k},\bsy{l}\rk}^{-1}\ \ud\bsy{x} \\
&\hspace{1cm} = -\frac{1}{4k^2}
\bma{ccc}
0 & 0 & 0 \\
0 & \eins & 0 \\
0 & 0 & \eins
\ema
-\frac{1}{4k^4} 
\bma{ccc}
0 & 0 & 0 \\
0 & V(\bsy{0}) & 0 \\
0 & 0 & V(\bsy{l})
\ema
\\
&\hspace{2cm} + \frac{1}{8 \ui k^5}
\bma{ccc}
0 & 0 & 0 \\
0 & V'(\bsy{0}) & 0 \\
0 & 0 & -V'(\bsy{l})
\ema
+ \Or(k^{-6})
}
\end{equation}
We still need to multiply this with the $\mfS$-matrix whose asymptotic expansion was 
determined in Theorem~\ref{79} and in Corollary~\ref{118}, 
\begin{equation}
\label{S-matrix}
\eq{
&\mfS\lk k;P,L\rk\\
&\hspace{0.5cm} = \mfS_{\infty}+2k^{-1}\ui L+2k^{-2}\lk \ui P^{\bot}\bsy{\beta}_1P-L^2\rk\\
&\hspace{1.0cm}+2k^{-3}\lk P^{\bot}\bsy{\beta}_1L-L\bsy{\beta}_1 P+\ui 
         P^{\bot}\bsy{\beta}_2P^{\bot}-\ui L^3\rk+\Or\lk k^{-4}\rk \\
&\hspace{0.5cm} = \mfS_{\infty} - \frac{2}{\ui k}L - \frac{1}{k^2}\lk  P^{\bot}
\bma{ccc}
0 & 0 & 0\\
0 & V(\bsy{0}) & 0\\ 
0 & 0 & V\lk\bsy{l}_{\bsy{e}}\rk
\ema
P + 2L^2\rk\\
&\hspace{1.0cm} - \frac{1}{\ui k^3}\lk P^{\bot}
\bma{ccc}
0 & 0 & 0\\
0 & V(\bsy{0}) & 0\\ 
0 & 0 & V\lk\bsy{l}_{\bsy{e}}\rk
\ema
L 
- L 
\bma{ccc}
0 & 0 & 0\\
0 & V(\bsy{0}) & 0\\ 
0 & 0 & V\lk\bsy{l}_{\bsy{e}}\rk
\ema
P  \right. \\
&\hspace{1.0cm} \left. + \frac{1}{2} P^{\bot} 
\bma{ccc}
0 & 0 & 0\\
0 & V'(\bsy{0}) & 0\\ 
0 & 0 & -V'\lk\bsy{l}_{\bsy{e}}\rk
\ema
P^{\bot} - 2 L^3\rk+\Or\lk k^{-4}\rk \\
}
\end{equation}
When taking the trace the third and the sixth term on the right-hand side vanish as they
contain $P$ and $P^{\bot}$, and $P$ and $L$, respectively. The latter case is due to 
$L=P^\bot L P^\bot$. 

This gives the following contribution of \eref{106} to the coefficients \eref{150},
\begin{equation*}
\eq{
 \tilde{b}_2 &= -\frac{1}{4}\tr\mfS_{\infty} \\
 \tilde{b}_3 &= \frac{1}{2 \ui} \tr L
}
\end{equation*}
\begin{equation}
\eq{
\tilde{b}_4 &= -\frac{1}{4}\tr \lk \mfS_{\infty}
\bma{ccc}
0 & 0 & 0 \\
0 & V(\bsy{0}) & 0 \\
0 & 0 & V(\bsy{l})
\ema 
\rk 
+ \frac{1}{2}\tr L^2 \\
&= -\frac{1}{4}\sum_{v \in \mc{V}} \sum_{e \sim v} \mf{S}_{\infty,ee}V_e(v) + \frac{1}{2}\tr L^2
}
\end{equation}
\begin{equation*}
\eq{
 \tilde{b}_5 &= \frac{1}{8\ui}\tr \lk \mfS_{\infty}
\bma{ccc}
0 & 0 & 0 \\
0 & V'(\bsy{0}) & 0 \\
0 & 0 & -V'(\bsy{l})
\ema
\rk \\ 
&\hspace{0.5cm} + \frac{1}{2\ui}\tr \lk L
\bma{ccc}
0 & 0 & 0 \\
0 & V(\bsy{0}) & 0 \\
0 & 0 & V(\bsy{l})
\ema
\rk-\frac{1}{2 \ui}\tr L^3\\
&\hspace{0.5cm} + \frac{1}{4\ui}\tr \lk P^{\bot}
\bma{ccc}
0 & 0 & 0\\
0 & V(\bsy{0}) & 0\\ 
0 & 0 & V(\bsy{l}_{\bsy{e}})
\ema
L \rk \\
&\hspace{0.5cm} +  \frac{1}{8 \ui} \tr \lk P^{\bot}
\bma{ccc}
0 & 0 & 0\\
0 & V'(\bsy{0}) & 0\\ 
0 & 0 & -V'(\bsy{l}_{\bsy{e}})
\ema
P^{\bot} \rk \\
& =  \frac{1}{8\ui}\sum_{v \in \mc{V}} \sum_{e \sim v} \mf{S}_{\infty,ee}V_{e}'(v) 
 + \frac{3}{4 \ui} \sum_{v \in \mc{V}} \sum_{e \sim v}  L_{ee} V_e(v)\\
&\hspace{0.5cm} - \frac{1}{2 \ui}\tr L^3
 - \frac{1}{8\ui}\sum_{v\in\mc{V}}\sum_{e\sim v}P^{\bot}_{ee}V_e'(v).
}
\end{equation*}
Combined with the contributions in \eref{iop1} this finally proves \eref{150}.
\end{proof}
\section{Asymptotics of the heat kernel}
\label{sec:heatasy}
We shall use the relation
\begin{equation}
\label{Ltransf}
\ue^{-Ht} -  J_{\ex}\ue^{\Delta_{D/N}t}J^\ast_{\ex} = \frac{\ui}{2\pi}\int_\gamma
\ue^{-\lambda t}\lk R_H(k^2)-J_{\ex}R_{D/N}(k^2)J^\ast_{\ex}\rk\ \ud\lambda,
\end{equation}
on the level of kernels, where the (positively oriented) contour $\gamma$ encloses 
$\sigma(H)$, to determine an asymptotic expansion for small $t$ of the trace of 
\eref{Ltransf}. This approach is based on the following lemma.
\begin{lemma}[\cite{Wong:2001}, p. 31]
\label{46}
Let $f\in C(0,\infty)$ such that $\ue^{-ct}f(t)\in L^1(0,\infty)$ for some $c\in\rz^+$. Let 
\begin{equation}
F(z) = \int_0^\infty\ue^{-zt}\,f(t)\ \ud t
\end{equation}
be the Laplace-transform of $f$. Assume that $F$ has a complete asymptotic expansion
\begin{equation}
\label{47}
F(z)\sim\sum\limits_{n=0}^{\infty}a_n\Gamma\lk d_n\rk z^{-d_n},\quad |z|\rightarrow \infty,
\end{equation}
uniformly in $\left|\arg(z-c)\right|\leq\frac{\pi}{2}$, such that $d_n\to\infty$ when 
$n\to\infty$. Then $f$ has a complete asymptotic expansion,
\begin{equation}
\label{48}
f(t)\sim\sum\limits_{n=0}^{\infty}a_nt^{d_n-1},\quad t\to 0^+.
\end{equation}
\end{lemma}
First, however, we establish the following.
\begin{prop}
\label{l49}
Let $\Gamma$ be a compact or non-compact metric graph with Schr\"odinger operator $H$. 
Then, for every $t>0$, the operator $\ue^{-H t}$ is an integral operator with kernel
\be
\label{lc1}
k_{H}(t;\cdot,\cdot)\in L^{\infty}(\Gamma\times\Gamma)\cap C^{\infty}(\Gamma\times\Gamma).
\ee
The heat kernel \eref{lc1} possesses a complete asymptotic expansion for $t\to 0^+$ in powers
of $\sqrt{t}$. On the diagonal the leading terms are:
\begin{itemize}
\item[(i)] $x$ not an edge end:
\be
\label{lc3}
k_{H,ee}\lk t; x,x\rk=\frac{1}{2\sqrt{\pi t}}\lk 1+\Or\lk{t}\rk\rk,
\ee
\item[(ii)] $x \cong (v,e)$ and $y \cong (v,e')$ with $e \neq e'$:
\be
\label{lc4}
k_{H,ee'}\lk t; x,y\rk=\frac{{\mfS_{\infty,xy}}}{2\sqrt{\pi t}}\lk 1+
2\sqrt{\pi t}L_{xy}+\Or\lk{t}\rk\rk,
\ee
\item[(iii)] $x \cong (v,e)$: 
\be
\label{lc5}
k_{H,ee}\lk t; x,x\rk=\frac{1+{\mfS_{\infty,xx}}}{2\sqrt{\pi t}}\lk 1+
2\sqrt{\pi t}L_{xx}+\Or\lk\sqrt{t}\rk\rk.
\ee
\end{itemize}
\end{prop}
\begin{proof}
The fact that $\ue^{-Ht}$ is an integral operator whose kernel satisfies
\eref{lc1} is proved in the same way as \cite[Lemma 6.1]{KostrykinSchrader:2006b}.
For this one needs to know that $\sigma(H)$ is bounded from below (Proposition~\ref{z8})
and that the resolvent kernel is composed of smooth functions (Theorem~\ref{22}
and Lemma~\ref{6}). 

The leading diagonal terms \eref{lc3}-\eref{lc5} follow from an application of 
Lemma~\ref{46} to Lemma~\ref{l1}.
\end{proof}
We are now in a position to determine heat kernels from resolvent kernels. As we
are eventually interested in trace asymptotics we need to consider the difference
\eref{Ltransf} of heat operators.
\begin{theorem}
\label{49}
Let $\Gamma$ be a compact or non-compact metric graph with Schr\"odinger operator $H$. 
Then, for every $t>0$, the difference \eref{Ltransf} of heat operators is a trace-class
integral operator. Its kernel $k_{H,D/N}(t;\cdot,\cdot)$ satisfies
\be
\label{c1}
k_{H,D/N}(t;\cdot,\cdot)\in L^{\infty}(\Gamma\times\Gamma)\cap C^{\infty}(\Gamma\times\Gamma).
\ee
The trace of \eref{Ltransf} possesses a complete asymptotic expansion for $t\to 0^+$ 
given by
\begin{equation}
\label{50}
\tr\lk \ue^{-Ht} -  J_{\ex}\ue^{\Delta_{D/N}t}J^\ast_{\ex}\rk 
  =\int_{\Gamma}\tr k_{H,D/N}(t;\bsy{x},\bsy{x})\ \ud\bsy{x}
  \sim\sum_{n=1}^{\infty}a_nt^{\frac{n}{2}-1},
\end{equation} 
where 
\begin{equation}
\label{51}
\eq{
a_{2n}&=\frac{(-1)^n}{(n-1)!}b_{2n},\quad n\in\nz\\
a_{2n+1}&=\ui\frac{(-1)^{n+1}}{\Gamma\lk n+\frac{1}{2}\rk}b_{2n+1}, \quad n\in\nz_0,
}
\end{equation}
with the coefficients $b_n$ from $\eref{150a}$. 
\end{theorem}
\begin{proof}
The fact that the difference \eref{Ltransf} of heat operators is trace class is proven
in complete analogy to the trace-class property of the corresponding difference of
resolvents, see Proposition~\ref{141b}. The integral kernel for \eref{Ltransf} is the 
difference of the heat kernel \eref{lc1} and the kernel for 
$J_{\ex}\ue^{\Delta_{D/N}t}{J_{\ex}}^{\ast}$ that is given in 
\cite{KostrykinSchrader:2006b,Schrader:2007}. Both kernel are in 
$L^{\infty}(\Gamma\times\Gamma)\cap C^{\infty}(\Gamma\times\Gamma)$, hence \eref{c1} follows.

As for the asymptotics, we apply Lemma~\ref{46} to
\be
\label{rk12}
\eq{
&\tr\lk R_{H}\lk -k^2\rk-J_{\ex}R_{D/N}\lk -k^2\rk J_{\ex}^{\ast}\rk \\
&\hspace{2cm} =
\int_{0}^{\infty}\ue^{-k^2t}\tr\lk \ue^{-Ht} -  J_{\ex}\ue^{\Delta_{D/N}t}J^\ast_{\ex}\rk\ \ud t, 
}
\ee 
where $k^2>-\inf\sigma(H)$. For this we first notice that
\be
\ue^{-ct}\tr\lk \ue^{-Ht} -  J_{\ex}\ue^{\Delta_{D/N}t}J^\ast_{\ex}\rk
\in L^1(0,\infty),
\ee
iff $c>-\inf\sigma(H)$. Lemma~\ref{46} requires the left-hand side of \eref{rk12} to
posses a complete asymptotic expansion for $|k^2|\to\infty$ in the right half-plane
$|\arg(k^2-c)|\leq\frac{\pi}{2}$. Such an expansion is indeed given by Theorem~\ref{151},
under the condition that $k^2\in S_{2\delta}$ which contains the half-plane
$|\arg(k^2-c)|\leq\frac{\pi}{2}$. Comparing \eref{47} and \eref{150} we identify 
$d_n=\frac{n}{2}$, leading to the relations \eref{51}.
\end{proof}
The first few heat-kernel coefficients can be worked out explicitly, making use of the 
resolvent coefficients \eref{150}. Using Theorems~\ref{151} and \ref{49} we can read 
off the first five coefficients. A direct calculation leads to
\be
\eq{
&\tr\lk \ue^{-Ht}-J_{\ex}\ue^{\Delta_{D/N}t}J^\ast_{\ex}\rk   \\
&\qquad =\frac{\mc{L}}{\sqrt{4\pi t}} + \frac{1}{4}\lk\tr\mf{S}_\infty\mp E_{\ex}\rk
   +\lk -\frac{1}{2}\int_{\Gamma_{\inte}}V(\bsy{x})\ \ud\bsy{x}+\tr L\rk\sqrt{t}\\
&\qquad\quad +\Or(t)\ .
}
\ee
an one can observe that the first two terms of this expansion are independent of the potential.
They indeed agree with the result \cite[Theorem 4.1]{Schrader:2007}, where the case of 
$V\equiv 0$ and vertex conditions such that $L=0$ is covered. From 
\cite[Theorem 4.1]{Schrader:2007} it also follows that $a_n=0$ for $n\geq 3$ in the case 
covered there. The terms involving the potential also agree with the ones computed in 
\cite{Ralf:2012a}. The presence of the potential, and the possibility of more general boundary
conditions allowing $L\neq 0$, implies that in general all heat-kernel coefficients $a_n$ 
will be non-zero.
\ack{S.E.\ acknowledges support by the German Academic Exchange Service (DAAD) and the 
German Research Foundation (DFG). This research was supported by the EPSRC network 
{\it Analysis on Graphs and Applications} (EP/1038217/1).}

\vspace*{1cm}

\vskip6pt
\bibliographystyle{amsalpha}

\bibliography{litver}
\end{document}

%% file: def.tex
\newcommand{\mfS}{\mathfrak{S}}

\newcommand{\ex}{\mathrm{ex}}
\newcommand{\inte}{\mathrm{int}}

\newcommand{\re}{\mathrm{Re}}

\newcommand{\ud}{\mathrm{d}}
\newcommand{\ui}{\mathrm{i}}
\newcommand{\ue}{\mathrm{e}}

\newcommand{\rz}{{\mathbb R}}
\newcommand{\nz}{{\mathbb N}}
\newcommand{\kz}{{\mathbb C}}

\newcommand{\eins}{\mathds{1}}
\newcommand{\be}{\begin{equation}}
\newcommand{\ee}{\end{equation}}

\newcommand{\lk}{\left(}
\newcommand{\rk}{\right)}
\newcommand{\lgk}{\left\{}
\newcommand{\rgk}{\right\}}
\newcommand{\mf}{\mathfrak}
\newcommand{\mc}{\mathcal}
\newcommand{\ba}{\begin{aligned}}
\newcommand{\ea}{\end{aligned}}
\newcommand{\eq}{\eqalign}

\newcommand{\bsy}{\boldsymbol}
\newcommand{\bma}{\left(
\begin{array}
}
\newcommand{\ema}{\end{array}
\right)}
\DeclareMathOperator{\im}{Im}

\newtheorem{theoremd}{Theorem and Definition}[section] 
\newtheorem{theorem}[theoremd]{Theorem} 
\newtheorem{lemma}[theoremd]{Lemma}
\newtheorem{prop}[theoremd]{Proposition}
\newtheorem{cor}[theoremd]{Corollary}
\newtheorem{defn}[theoremd]{Definition}
\newtheorem{rem}[theoremd]{Remark}